\documentclass[12pt]{article}

\usepackage{acro}
\usepackage{amsmath}
\usepackage{amssymb}
\usepackage{amsthm}
\usepackage{bm}
\usepackage{bbm}
\usepackage{color}
\usepackage{enumitem}
\usepackage{fullpage}
\usepackage{graphicx}
\usepackage{hyperref}
\usepackage{natbib}
\usepackage{cleveref}

\newtheorem{assumption}{Standing Assumption}
\newtheorem{theorem}{Theorem}
\newtheorem{definition}{Definition}
\newtheorem{remark}{Remark}
\newtheorem{proposition}{Proposition}
\newtheorem{example}{Example}
\newtheorem{lemma}{Lemma}
\newtheorem{exercise}{Exercise}

\DeclareAcronym{mmd}{short = MMD, long = maximum mean discrepancy}
\DeclareAcronym{rkhs}{short = RKHS, long = reproducing kernel Hilbert space}
\DeclareAcronym{qmc}{short = QMC, long = quasi Monte Carlo}
\DeclareAcronym{pdf}{short = PDF, long = probability density function}
\DeclareAcronym{ksd}{short = KSD, long = kernel Stein discrepancy}
\DeclareAcronym{mcmc}{short = MCMC, long = Markov chain Monte Carlo}
\DeclareAcronym{smc}{short = SMC, long = sequential Monte Carlo}

\DeclareMathOperator*{\argmin}{arg\,min}

\begin{document}

\title{Minimum Discrepancy Methods in Uncertainty Quantification\footnote{The lectures were prepared for the \'{E}cole Th\'{e}matique sur les Incertitudes en Calcul Scientifique (ETICS) in September 2021.}}
\author{Chris. J. Oates\footnote{Correspondence should be sent to \texttt{chris.oates@ncl.ac.uk}.} \\
\small Newcastle University}

\maketitle

\tableofcontents

\vspace{20pt}

These lectures concern the discrete approximation of objects that are in some sense infinite-dimensional.
This problem is ubiquitous to numerical computation in general.
Specifically, we will consider discrete approximation of probability distributions $P$ that may be defined on an infinite set, such as $[0,1]^d$ or $\mathbb{R}^d$.
See \Cref{fig: Graf}.
The basic questions here are: (1) how many states are required to achieve a given level of approximation? (2) how can such approximations be constructed?

\begin{figure}[t!]
\centering
\includegraphics[width = \textwidth]{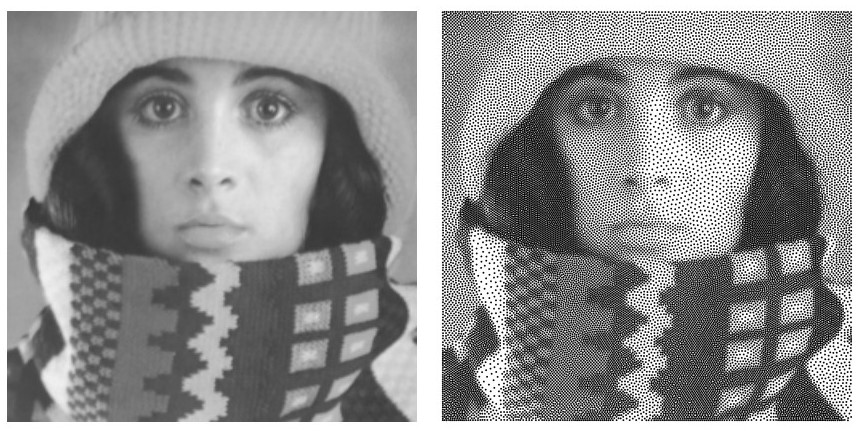}
\caption{\textit{Quantisation} is the act of approximating a probability distribution $P$, defined on an infinite set, with a discrete distribution $\sum_{i=1}^n w_i \delta(\bm{x}_i)$ supported on a finite set of states $\{\bm{x}_i\}_{i=1}^n$. 
Here $\delta(\bm{x}_i)$ denotes a Dirac distribution centred at $\bm{x}_i$.
The left image can be considered to represent a \ac{pdf} $p(\bm{x})$ for $P$, defined for $\bm{x} \in [0,1]^d$, and the dots in the right image can be considered to represent $\bm{x}_i$, with the size of the does representing $w_i$.
This image is due to \citet{graf2012quadrature}.}
\label{fig: Graf}
\end{figure}

\section{Classical Discrepancy Theory}

Here we start with some motivation from a classical perspective, which considers approximation of the uniform distribution $P$ on $[0,1]^d$ by an (un-weighted) collection of states $\{\bm{x}_i\}_{i=1}^n$.
For $\bm{x} \in [0,1]^d$, let $[\bm{0},\bm{x}) = [0,x_1) \times \dots \times [0,x_d)$.

\begin{definition}[Local discrepancy]
The \emph{local discrepancy} of a collection of states $\bm{x}_1,\dots,\bm{x}_n \in [0,1]^d$ at $\bm{a} \in [0,1]^d$ is
$$
\Delta(\bm{a}) = \Delta(\bm{a};\bm{x}_1,\dots,\bm{x}_n) = \frac{1}{n}\sum_{i=1}^n \mathbbm{1}_{[\bm{0},\bm{a})}(\bm{x}_i ) - \prod_{j=1}^d a_j .
$$
\end{definition}

\begin{definition}[Star discrepancy]
The \emph{star discrepancy} of $\bm{x}_1,\dots,\bm{x}_n \in [0,1]^d$ is
$$
D_n^* = D_n^*(\bm{x}_1,\dots,\bm{x}_n) = \sup_{\bm{a} \in [0,1)^d} |\Delta(\bm{a}; \bm{x}_1,\dots,\bm{x}_n)| .
$$
\end{definition}

\noindent See the illustration in \Cref{fig: star}.

\begin{figure}[t!]
\centering
\includegraphics[width = 0.8\textwidth]{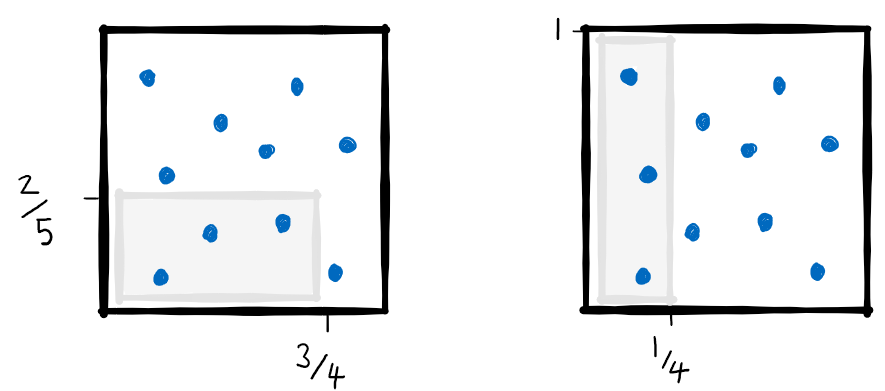}
\caption{\textit{Star discrepancy} is defined as the largest difference between the volume of a hyper-rectangle $[\bm{0},\bm{a})$ and the proportion of states $\bm{x}_i$ (blue dots) which are contained in the hyper-rectangle, the so-called \textit{local discrepancy} at $\bm{a} \in [0,1]^d$.
In the left hand panel, the proportion of states in the shaded hyper-rectangle is exactly equal to the volume of the hyper-rectangle, while this is not the case in the right hand panel.
}
\label{fig: star}
\end{figure}

\begin{remark}
In dimension $d=1$ is is clear that a regular grid $\bm{x}_1 = 0$, $\bm{x}_2 = 1/(n-1)$, $\dots$, $\bm{x}_n = 1$ minimises $D_n^*(\bm{x}_1,\dots,\bm{x}_n)$.
In this univariate setting, you may recognise that the star discrepancy from the \textit{Kolmogorov--Smirnov} uniformity test.
\end{remark}

\begin{remark}
It is simplifies discussion to \textit{anchor} the hyper-rectangle at $\bm{0}$, but one can also consider an alternative to star discrepancy with hyper-rectangles $[\bm{a},\bm{b})$ that are \textit{un-anchored}.
This alternative discrepancy takes values in $[D_n^*,2^d D_n^*]$, so in terms of fixed $d$ asymptotics its behaviour is identical.
\end{remark}

\begin{remark}
It may surprise you how little is known about star discrepancy.
In dimensions $d \leq 2$, it has been proven that there exist a constant $0 < C < \infty$ such that, for any choice of $\bm{x}_1,\dots,\bm{x}_n$ and $n \in \mathbb{N}$,
$$
C \frac{(\log n)^{d-1}}{n} \leq D_n^*(\bm{x}_1,\dots,\bm{x}_n),
$$
but for dimension $d > 2$ this bound is only conjectured.
\end{remark}

\begin{remark}
There are numerous algorithms which aim to generate collections of states with small star discrepancy, including the Halton sequence, the Hammersley set, Sobol sequences, and various non-independent sampling methods, which are sometimes collectively referred to as \emph{\ac{qmc}} methods.
For some of these methods it is known that, for an appropriate constant $0 < C < \infty$ and subsequence of values $n$ in $\mathbb{N}$,
$$
D_n^*(\bm{x}_1,\dots,\bm{x}_n) \leq C \frac{(\log n)^d}{n},
$$
meaning that the rate of convergence in $n$ is close to the conjectured optimal rate.
\ac{qmc} will not be discussed further, because our aim in these lectures is to deal with more general probability distributions $P$ that arise in applications of uncertainty quantification.
\end{remark}

\begin{remark}
A regular grid consisting of of $n = m^d$ states does \emph{not} minimise star discrepancy in dimension $d > 1$.
Indeed, for a regular grid (i.e. a $d$ dimensional Cartesian product of regular grids over the unit interval), one can show that 
$$
\frac{1}{2n^{1/d}} \leq D_n^* \leq \frac{d}{2n^{1/d}} .
$$
See \citet[][Remark 2.20]{leobacher2014introduction}.
\end{remark}

At this point it should be clear that even the simplest quantisation problems can be far from trivial.
The aim of the next section is to relate the slightly abstract notion of quantisation to concrete problems of numerical integration.

\paragraph*{Chapter Notes}

The presentation of star discrepancy followed Chapter 15 of \citet{mcbook}, which is currently a freely available online textbook.
The same reference provides an excellent introduction to \ac{qmc}.

\section{Numerical Cubature}

One of the most basic operations that one could hope to perform with a probability distribution is to compute expectations of random variables; i.e. to compute integrals of the form $\int f \mathrm{d}P$, or $\int f(\bm{x}) p(\bm{x}) \mathrm{d}\bm{x}$ if the probability distribution $P$ admits a \ac{pdf} $p(\bm{x})$.
In general such integrals do not possess a closed form and numerical integration (also called \textit{cubature}) will be required.
Quantisation is useful for cubature, since we can replace $P$ with a discrete approximation $\sum_{i=1}^n w_i \delta(\bm{x}_i)$ to obtain a closed form numerical approximation $\sum_{i=1}^n w_i f(\bm{x}_i)$ to the integral.
Approximations of this form are sometimes called \textit{cubature rules}.
In this section we will see how star discrepancy can be used to analyse the accuracy of cubature rules in the case where $P$ is a uniform distribution on $[0,1]^d$.

\subsection{Koksma--Hlawka Inequality}

Let $\bm{x}_{\mathfrak{u}}$ denote the components of a $d$-dimensional vector $\bm{x}$ that are indexed by the set $\mathfrak{u} \subseteq \{1,\dots,d\}$.
The shorthand $f(\bm{x}_{\mathfrak{u}},\bm{1})$ will be used to represent $f(\bm{y})$, where the vector $\bm{y}$ is defined by $y_i = x_i$ if $i \in \mathfrak{u}$ and otherwise $y_i = 1$.
The mixed partial derivative of $f$ with respect to each of the co-ordinates in $\mathfrak{u}$ is denoted $\partial^{|\mathfrak{u}|}f/\partial \bm{x}_{\mathfrak{u}}$, again for $\mathfrak{u} \subseteq \{1,\dots,d\}$.

\begin{definition}[Variation] \label{def: var}
For $f : [0,1]^d \rightarrow \mathbb{R}$ with continuous mixed partial derivatives, we define the \emph{variation} of $f$ to be
$$
\|f\|_1 = \sum_{\mathfrak{u} \subseteq \{1,\dots,d\} } \int_{[0,1]^{|\mathfrak{u}|}} \left| \frac{\partial^{|\mathfrak{u}|}f}{\partial \bm{x}_{\mathfrak{u}}}(\bm{x}_{\mathfrak{u}} , \bm{1}) \right| \mathrm{d}\bm{x} ,
$$
where the sum runs over all $2^d$ subsets $\mathfrak{u} \subseteq \{1,\dots,d\}$.
\end{definition}

The term ``variation'' is overloaded in the literature and we use it only informally to give a name to the norm defined in \Cref{def: var}.
The concept of variation enables the accuracy of cubature rules to be analysed:

\begin{theorem}[Koksma--Hlawka inequality] \label{thm: KH}
Let $f : [0,1]^d \rightarrow \mathbb{R}$ have continuous mixed partial derivatives.
Then
\begin{equation}
\left| \frac{1}{n} \sum_{i=1}^n f(\bm{x}_i) - \int_{[0,1]^d} f(\bm{x}) \mathrm{d}\bm{x} \right| \leq \|f\|_1 D_n^*(\bm{x}_1,\dots,\bm{x}_n) . \label{eq: KH}
\end{equation}
\end{theorem}

\begin{remark}
The term $\|f\|_1$ in \eqref{eq: KH} quantifies the complexity of the integrand, while the star discrepancy quantifies the suitability of the states $\{\bm{x}_i\}_{i=1}^n$.
Thus the quality of the set $\{\bm{x}_i\}_{i=1}^n$ as a quantisation of $P$ controls the accuracy of the cubature rule.
\end{remark}

\begin{remark}
The $f(\bm{1})$ term from $\|f\|_1$ can actually be removed, since the left hand side of \eqref{eq: KH} is invariant to $f(\bm{1})$, meaning that we can apply \eqref{eq: KH} to the function $\bm{x} \mapsto f(\bm{x}) - f(\bm{1})$ instead, to obtain a tighter bound. 
If we do that, then we obtain (up to small technicalities) the original formulation of Koksma--Hlawka.
\end{remark}

Later we will prove \Cref{thm: KH} in full, but for now we aim to present an elementary proof of \Cref{thm: KH} in the case $d=1$.
For this we need the following:

\begin{lemma} \label{lem: Owen 15.1}
Let $f : [0,1] \rightarrow \mathbb{R}$ be continuously differentiable and let $x_1,\dots,x_n \in [0,1]$.
Then
$$
\frac{1}{n} \sum_{i=1}^n f(x_i) - \int_0^1 f(x) \mathrm{d}x = - \int_0^1 \Delta(x) f'(x) \mathrm{d}x
$$
where $\Delta(x) = \Delta(x;x_1,\dots,x_n)$ is the local discrepancy.
\end{lemma}
\begin{proof}
Note that the regularity assumption allows us to write
\begin{equation}
f(x) = f(1) - \int_x^1 f'(y) \mathrm{d}y . \label{eq: calc id}
\end{equation}
Substituting \eqref{eq: calc id} into the expression for the cubature error, we obtain
\begin{align*}
\frac{1}{n} \sum_{i=1}^n f(x_i) - \int_0^1 f(x) \mathrm{d}x & = \int_0^1 \int_x^1 f'(y) \mathrm{d}y \mathrm{d}x - \frac{1}{n} \sum_{i=1}^n \int_{x_i}^1 f'(y) \mathrm{d}y \\
& = \int_0^1 \int_0^y f'(y) \mathrm{d}x \mathrm{d}y - \int_0^1 \frac{1}{n} \sum_{i=1}^n \mathbbm{1}_{(x_i,1]}(y) f'(y) \mathrm{d}y \\
& = \int_0^1 f'(y) \underbrace{ \left[ y - \frac{1}{n} \sum_{i=1}^n \mathbbm{1}_{(x_i,1]}(y) \right] }_{= - \Delta(y)} \mathrm{d}y
\end{align*}
as required.
\end{proof}

\begin{proof}[Proof of \Cref{thm: KH} ($d=1$)]
From \Cref{lem: Owen 15.1} we have
$$
\left| \frac{1}{n} \sum_{i=1}^n f(x_i) - \int_0^1 f(x) \mathrm{d}x \right| \leq \int_0^1 \left| \Delta(x) f'(x) \right| \mathrm{d}x \leq \underbrace{ \sup_{x \in [0,1]} |\Delta(x)| }_{ = D_n^*} \underbrace{ \int_0^1 \left| f'(x) \right| \mathrm{d}x }_{= \|f\|_1 - |f(1)| } ,
$$
as required.
\end{proof}

\paragraph*{Chapter Notes}

See Remark 2.19 in \citet{dick2010digital} for a detailed discussion of how \Cref{thm: KH} relates to the original formulation of Koksma and Hlawka.
\Cref{lem: Owen 15.1} and the proof of \Cref{thm: KH} for $d=1$ can be found in Section 2.2 in \citet{dick2010digital}; see also Theorem 15.1 in \citet{mcbook}.
The one-dimensional version of the Koksma--Hlawka inequality is sometimes called \textit{Koksma's inequality} or \textit{Zaremba's identity} \citep[][p18]{dick2010digital}.

\subsection{Cubature Error Representer}
\label{subsec: representers}

In this section and the next, we will introduce the mathematical tools that are needed to prove \Cref{thm: KH} in full.
These tools will also be useful later, when we consider practical algorithms for quantisation of general probability distributions $P$.
The aim is to generalise the concept of variation in \Cref{def: var}, to allow for functions $f$ of different regularity to be integrated.

The basic idea is as follows: we consider the set $\mathcal{S}(k)$ of all functions of the form $f(\bm{x}) = \sum_{i=1}^m b_i k(\bm{x} , \bm{y}_i)$, where $k$ is to be specified, the $\bm{y}_i$ are fixed states, and $n \in \mathbb{N}$.
The function $k$ determines the regularity of the elements in $\mathcal{S}(k)$; for example, if $k(\bm{x},\bm{y}) = \exp(-\|\bm{x} - \bm{y}\|)$ then the elements of $\mathcal{S}(k)$ are continuous but not differentiable, while if $k(\bm{x},\bm{y}) = \exp(-\|\bm{x} - \bm{y}\|^2)$ then the elements of $\mathcal{S}(k)$ are infinitely differentiable.
See \Cref{fig: rkhs}.
Since we aim to perform mathematical analysis, we will want to endow the set $\mathcal{S}(k)$ with mathematical structure that we can exploit.
It is clearly a vector space (over the reals) of functions when (pointwise) addition and scalar multiplication are defined.
In addition to that, we will want to make use of an inner product
$$
\langle f , g \rangle_{\mathcal{S}(k)} = \sum_{i=1}^m \sum_{j=1}^n b_i c_j k(\bm{y}_i , \bm{z}_j) , \quad f(\bm{x}) = \sum_{i=1}^m b_i k(\bm{x} , \bm{y}_i), \quad g(\bm{x}) = \sum_{j=1}^n c_j k(\bm{x} , \bm{z}_j) ,
$$
for which we must require that $k$ is \textit{symmetric} (i.e. $\langle f , g \rangle_{\mathcal{S}(k)} = \langle g , f \rangle_{\mathcal{S}(k)}$) and \textit{positive definite} (i.e. $\langle f , f \rangle_{\mathcal{S}(k)} > 0$ for all $f \neq 0$).
This inner product is useful because it satisfies a \textit{reproducing property}, meaning that 
\begin{align*}
\langle f , k(\cdot,\bm{x}) \rangle_{\mathcal{S}(k)} = f(\bm{x}) , 
\end{align*}
and suggesting the formal manipulation
\begin{align*}
\frac{1}{n} \sum_{i=1}^n f(\bm{x}_i) - \int_{[0,1]^d} f(\bm{x}) \mathrm{d}\bm{x} & = \frac{1}{n} \sum_{i=1}^n \langle f , k(\cdot,\bm{x}_i) \rangle_{\mathcal{S}(k)} - \int_{[0,1]^d} \langle f , k(\cdot,\bm{y}) \rangle_{\mathcal{S}(k)} \mathrm{d}\bm{y} \\
& \stackrel{?}{=} \Bigg\langle f , \underbrace{ \frac{1}{n} \sum_{i=1}^n k(\cdot,\bm{x}_i) - \int_{[0,1]^d} k(\cdot,\bm{y}) \mathrm{d}\bm{y} }_{= e(\cdot) } \Bigg\rangle_{\mathcal{S}(k)} ,
\end{align*}
where $e(\cdot)$ is referred to as the \textit{representer} of the \emph{cubature error}.
The representer $e(\cdot)$ would completely characterise the error of the (un-weighted) cubature rule based on the states $\{\bm{x}_i\}_{i=1}^n$.
For example, if $e(\bm{x}) = 0$ for all $\bm{x}$ then the cubature rule would be exact for all integrands $f \in \mathcal{S}(k)$.

Unfortunately the inner product space $\mathcal{S}(k)$ is not \textit{complete}, in the sense that limits of functions of the form $f(\bm{x}) = \sum_{i=1}^m b_i k(\bm{x} , \bm{y}_i)$ need not be elements of the set $\mathcal{S}(k)$.
In particular, the integral $\int_{[0,1]^d} k(\cdot,\bm{y}) \mathrm{d}\bm{y}$ is not an element of $\mathcal{S}(k)$, meaning that the formal manipulation above is not well-defined.
For this technical reason we work with a larger set $\mathcal{H}(k)$, called the \textit{completion} of $\mathcal{S}(k)$, whose definition we present next.

\begin{figure}[t!]
\centering
\includegraphics[width = 0.4\textwidth]{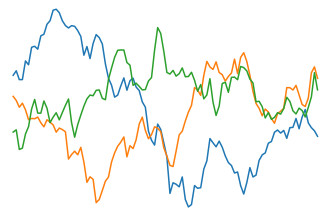} \hspace{20pt}
\includegraphics[width = 0.4\textwidth]{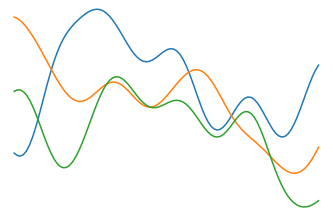}
\caption{\textit{\Acp{rkhs}:}
The left panel represents elements from a \ac{rkhs} whose kernel is non-differentiable, while the right panel corresponds to an infinitely differentiable kernel. }
\label{fig: rkhs}
\end{figure}

\subsection{Reproducing Kernel Hilbert Spaces}

The main mathematical tool that we will exploit is that of a \textit{reproducing kernel}:

\begin{definition}[Reproducing kernel Hilbert space] \label{def: rkhs}
Let $\mathcal{X}$ be a set and consider a symmetric and positive definite function $k : \mathcal{X} \times \mathcal{X} \rightarrow \mathbb{R}$.
Then a \emph{\ac{rkhs}} with \emph{reproducing kernel} (or simply \emph{kernel}) $k$ is an inner product space $\mathcal{H}(k)$ of functions $f : \mathcal{X} \rightarrow \mathbb{R}$, such that
\begin{enumerate}
\item $k(\cdot,\bm{x}) \in \mathcal{H}(k)$ for all $\bm{x} \in \mathcal{X}$
\item $\langle f , k(\cdot,\bm{x}) \rangle_{\mathcal{H}(k)} = f(\bm{x})$ for all $\bm{x} \in \mathcal{X}$ and all $f \in \mathcal{H}(k)$.
\end{enumerate}
\end{definition}

\begin{remark}
Given a symmetric positive definite function $k$, it can be shown that there exists a unique \ac{rkhs} $\mathcal{H}(k)$.
Conversely, each \ac{rkhs} admits a unique reproducing kernel, and that kernel is symmetric and positive definite.
\end{remark}

In general it is difficult to characterise the inner product induced by a reproducing kernel, and hence the elements of the \ac{rkhs}.
However, there are a number of important cases where this can be carried out:

\begin{example} \label{eq: finite dim rkhs}
The linear span of a finite collection of functions $e_1(\bm{x}), \dots, e_p(\bm{x})$ can be endowed with the structure of an \ac{rkhs} with reproducing kernel $k(\bm{x},\bm{y}) = \sum_{i=1}^p e_i(\bm{x}) e_i(\bm{y})$.
The induced inner product is $\langle f , g \rangle_{\mathcal{H}(k)} = b_1 c_1 + \dots + b_p c_p$, where $f(\bm{x}) = \sum_{i=1}^n b_i e_i(\bm{x})$ and $g(\bm{x}) = \sum_{i=1}^n c_i e_i(\bm{x})$.
\end{example}

\begin{example} \label{lem: Hilbert space}
The kernel 
\begin{equation}
k(\bm{x},\bm{y}) = \prod_{i=1}^d \left( 1 + \min(1-x_i,1-y_i) \right), \qquad \bm{x}, \bm{y} \in [0,1]^d \label{eq: kernel for KH}
\end{equation}
reproduces a Hilbert space with inner product
\begin{equation}
\langle f , g \rangle_{\mathcal{H}(k)} = \sum_{\mathfrak{u} \subseteq \{1,\dots,d\}} \int_{[0,1]^s} \frac{\partial^{|\mathfrak{u}|}f}{\partial \bm{x}_{\mathfrak{u}}}(\bm{x}_{\mathfrak{u}} , \bm{1}) \frac{\partial^{|\mathfrak{u}|}g}{\partial \bm{x}_{\mathfrak{u}}}(\bm{x}_{\mathfrak{u}} , \bm{1}) \mathrm{d}\bm{x}_{\mathfrak{u}} . \label{eq: inner prod}
\end{equation}
\end{example}

\begin{assumption} \label{ass: cts}
For all reproducing kernels $k$ considered in the sequel, we assume that $(\bm{x},\bm{y}) \mapsto k(\bm{x},\bm{y})$ is a continuous function. 
\end{assumption}

\begin{remark} \label{lem: cub er rep}
Identical manipulation to that presented in \Cref{subsec: representers} shows that the (Riesz) representer of the cubature error is
$$
e(\bm{x}) = e(\bm{x}; \bm{x}_1,\dots,\bm{x}_n)  = \frac{1}{n} \sum_{i=1}^n k(\bm{x},\bm{x}_i) - \int_{[0,1]^d} k(\bm{x},\bm{y}) \mathrm{d}\bm{y} .
$$
Note that the integral of the kernel is well-defined from \Cref{ass: cts}.
The interchange of integral and inner product in \Cref{subsec: representers} requires justification; this will be provided later in \Cref{lem: kme}.
\end{remark}

Armed with reproducing kernels and the cubature error representer, we can now prove \Cref{thm: KH} in full.
In what follows we let $\bm{x}_{i,j}$ denote the $j$th coordinate of the vector $\bm{x}_i$ and we let $\bm{x}_{i,\mathfrak{u}}$ denote the components of the vector $\bm{x}_i$ that are indexed by the set $\mathfrak{u} \subseteq \{1,\dots,d\}$.

\begin{proof}[Proof of \Cref{thm: KH}]
Consider the kernel $k$ in \eqref{eq: kernel for KH};
$$
k(\bm{x},\bm{y}) = \prod_{i=1}^d \left( 1 + \min(1-x_i,1-y_i) \right) 
$$
and compute the cubature error representer
\begin{align*}
e(\bm{x}) & = \frac{1}{n} \sum_{i=1}^n k(\bm{x},\bm{x}_i) - \int_{[0,1]^d} k(\bm{x},\bm{y}) \mathrm{d}\bm{y}  \\
& = \frac{1}{n} \sum_{i=1}^n \prod_{j=1}^d \left( 1 + \min(1-x_i,1-\bm{x}_{i,j}) \right) - \prod_{i=1}^d \frac{3 - x_i^2}{2} .
\end{align*}
Then, for $\mathfrak{u} \subseteq \{1,\dots,d\}$ and $(\bm{x}_{\mathfrak{u}},\bm{1}) \notin \{ \bm{x}_1,\dots,\bm{x}_n\}$,
\begin{align}
\frac{\partial^{|\mathfrak{u}|}e}{\partial \bm{x}_{\mathfrak{u}}}(\bm{x}_{\mathfrak{u}} , \bm{1}) & = (-1)^{|\mathfrak{u}|} \underbrace{ \left( \frac{1}{n} \sum_{i=1}^n \mathbbm{1}_{[\bm{0}_{\mathfrak{u}} , \bm{x}_{\mathfrak{u}})}(\bm{x}_{i,\mathfrak{u}}) - \prod_{i \in \mathfrak{u}} x_i \right) }_{ = \Delta(\bm{x}_{\mathfrak{u}}, \bm{1}) }  \label{eq: deriv h}  .
\end{align}
The assumed regularity of $f$ ensures that $f \in \mathcal{H}(k)$.
Plugging  \eqref{eq: deriv h} into \eqref{eq: inner prod}, we obtain
\begin{align*}
\frac{1}{n} \sum_{i=1}^n f(\bm{x}_i) - \int_{[0,1]^d} f(\bm{x}) \mathrm{d}\bm{x} & = \langle f , e \rangle_{\mathcal{H}(k)} \\
& = \sum_{\mathfrak{u} \subseteq \{1,\dots,d\}} (-1)^{|\mathfrak{u}|} \int_{[0,1]^d} \frac{\partial^{|\mathfrak{u}|}f}{\partial \bm{x}_{\mathfrak{u}}}(\bm{x}_{\mathfrak{u}} , \bm{1})  \Delta(\bm{x}_{\mathfrak{u}}, \bm{1}) \mathrm{d}\bm{x}_{\mathfrak{u}} .
\end{align*}
Finally, taking absolute values and using the bound $|\Delta(\bm{x}_{\mathfrak{u}}, \bm{1})| \leq \sup_{\bm{x} \in [0,1]^d} |\Delta(\bm{x})| = D_n^*$, we arrive at \eqref{eq: KH}.
\end{proof}

\paragraph*{Chapter Notes}

There are several excellent introductions to the theory of reproducing kernels, including \citet{wendland2004scattered,berlinet2011reproducing}.
The presentation of the Koksma--Hlawka inequality from a reproducing kernel perspective follows Section 2.4 in \citet{dick2010digital}.
A proof of the Koksma--Hlawka inequality, which does not require reproducing kernels, can be found as Theorem 5.5 in Chapter 2 of \citet{kuipers2012uniform}.
\Cref{lem: Hilbert space} can be found in Section 2.4 of \citet{dick2010digital}.

\section{Maximum Mean Discrepancy} \label{sec: mmd}

Now we move beyond the uniform distribution $P = \mathcal{U}([0,1]^d)$ and consider general probability distributions $P$ on general\footnote{These lecture notes deliberately avoid discussion of measure theory, but to be fully rigorous we restrict attention to Borel measures $P$ defined on a topological space $\mathcal{X}$.} domains $\mathcal{X}$.
The aim is to present a modern treatment of discrepancy and cubature error, that generalises the previous results.

\begin{definition}[Kernel mean embedding]
For a kernel $k$ and a probability distribution $P$, we call $\mu_P = \int k(\cdot,\bm{x}) \mathrm{d}P(\bm{x})$ the \emph{kernel mean embedding} of $P$ in $\mathcal{H}(k)$, whenever it is well-defined (see \Cref{lem: kme}).
\end{definition}

\begin{figure}
\centering
\includegraphics[width = 0.8\textwidth]{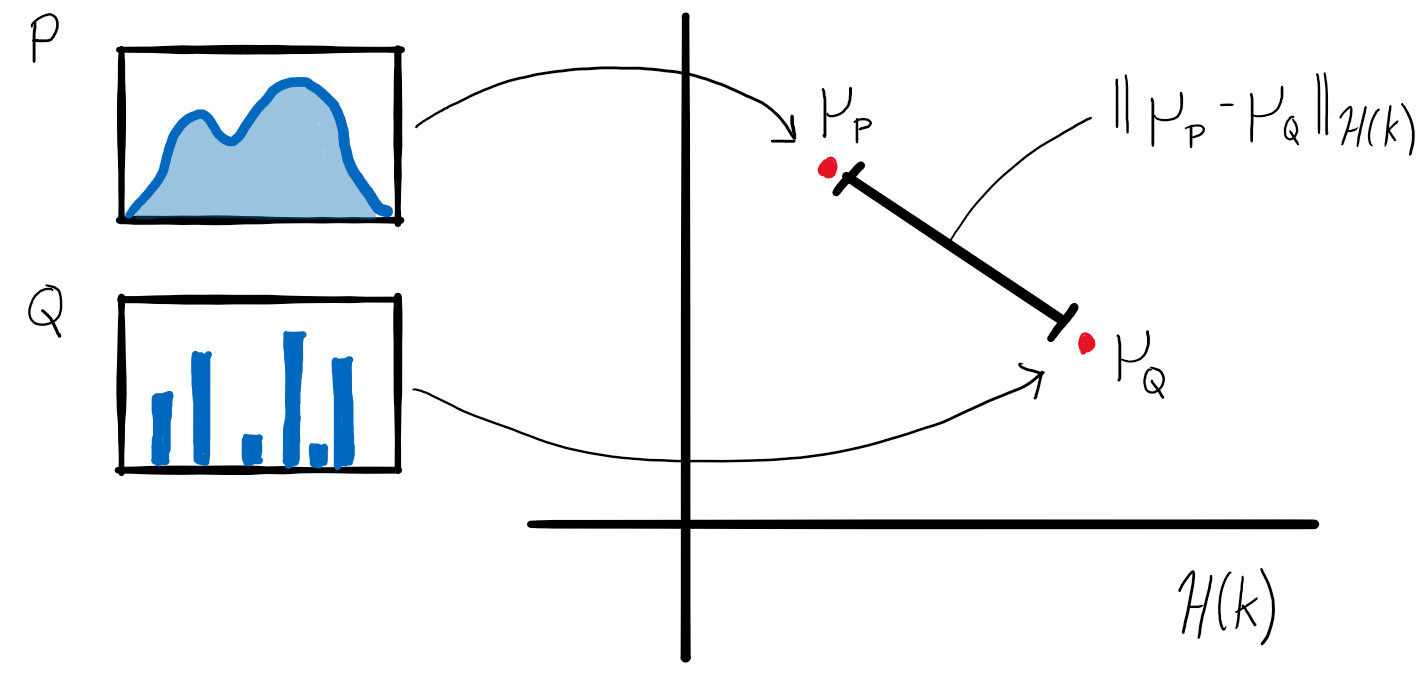}
\caption{\textit{Kernel mean embedding:} Two probability distributions $P$ and $Q$ are mapped to their respective elements $\mu_P$ and $\mu_Q$ in the \ac{rkhs} $\mathcal{H}(k)$. The distance (in $\mathcal{H}(k)$) between these kernel mean embeddings $\mu_P$ and $\mu_Q$ is called the \ac{mmd} between $P$ and $Q$.}
\label{fig: kme}
\end{figure}

\begin{lemma} \label{lem: kme}
If $\int \sqrt{k(\bm{x},\bm{x})} \mathrm{d}P(\bm{x}) < \infty$ then $\mu_P(\bm{x}) \in \mathcal{H}(k)$ and $\int f(\bm{x}) \mathrm{d}P(\bm{x}) = \langle f , \mu_P \rangle_{\mathcal{H}(k)}$.
\end{lemma}
\begin{proof}
Consider the linear operator $Lf = \int f(\bm{x}) \mathrm{d}P(\bm{x})$.
Then
\begin{align}
|Lf| = \left| \int f(\bm{x}) \mathrm{d}P(\bm{x}) \right| & \leq \int |f(\bm{x})| \mathrm{d}P(\bm{x})  \label{eq: L1} \\
& = \int | \langle f , k(\cdot,\bm{x}) \rangle_{\mathcal{H}(k)} | \mathrm{d}P(\bm{x}) \label{eq: L2} \\
& \leq \int \|f\|_{\mathcal{H}(k)} \| k(\cdot,\bm{x}) \|_{\mathcal{H}(k)} | \mathrm{d}P(\bm{x}) \label{eq: L3} \\
& = \int \sqrt{k(\bm{x},\bm{x})} \mathrm{d}P(\bm{x})  \|f\|_{\mathcal{H}(k)}  \nonumber
\end{align}
where \eqref{eq: L1} is Jensen's inequality, \eqref{eq: L2} is the reproducing property, and \eqref{eq: L3} is Cauchy--Schwarz.
This shows that $L$ is a \textit{bounded linear operator} from $\mathcal{H}(k)$ to $\mathbb{R}$.
Thus, from the Riesz representation theorem, there exists $h \in \mathcal{H}(k)$ such that $Lf = \langle f , h \rangle_{\mathcal{H}(k)}$.
Taking $f(\bm{x}) = k(\bm{y},\bm{x})$ and using the reproducing property leads to $\int k(\bm{y},\bm{x}) \mathrm{d}P(\bm{x}) = Lf = \langle f , h \rangle_{\mathcal{H}(k)} = h(\bm{y})$, so that $h = \int k(\cdot,\bm{x}) \mathrm{d}P(\bm{x})$, and so $\mu_p = h \in \mathcal{H}(k)$ with $Lf = \langle f , \mu_P \rangle_{\mathcal{H}(k)}$, as was claimed.
\end{proof}

\begin{assumption} \label{ass: trace class}
For all reproducing kernels $k$ and probability distributions $P$ considered in the sequel, we assume that $\int \sqrt{k(\bm{x},\bm{x})} \mathrm{d}P(\bm{x}) < \infty$.
\end{assumption}

In this more general setting the Riesz representer of the cubature error is the difference $e = \mu_{Q_n} - \mu_P$ of two kernel mean embeddings, where $Q_n = \sum_{i=1}^n w_i \delta(\bm{x}_i)$ is the discrete distribution on which the cubature rule is based.
i.e.
\begin{align*}
\sum_{i=1}^n w_i f(\bm{x}_i) - \int f \mathrm{d}P = \langle f , \mu_{Q_n} - \mu_P \rangle_{\mathcal{H}(k)} , \quad \mu_{Q_n} = \sum_{i=1}^n w_i k(\cdot, \bm{x}_i), \quad \mu_P = \int k(\cdot,\bm{x}) \mathrm{d}P(\bm{x}) .
\end{align*}

There are several different ways to systematically assess the performance of a cubature rule, but here we focus on a \textit{worst case} assessment:

\begin{definition}[Maximum mean discrepancy]
The \emph{\ac{mmd}} between two distributions $P$ and $Q$ is
\begin{align*}
D_k(P,Q) 
& = \sup_{\|f\|_{\mathcal{H}(k)} \leq 1} \left| \int f \mathrm{d}P - \int f \mathrm{d}Q  \right| ,
\end{align*}
also called the \emph{worst case cubature error} in the unit ball of $\mathcal{H}(k)$.
\end{definition}

A similar argument to \Cref{lem: cub er rep} shows that:

\begin{lemma} \label{lem: mmd ad embed}
$D_k(P,Q) = \|\mu_P - \mu_Q \|_{\mathcal{H}(k)}$.
\end{lemma}
\begin{proof}
Since $e = \mu_P - \mu_Q$ is the Riesz representer of the intergal approximation error, we may apply Cauchy--Schwarz to obtain
$$
\left| \int f \mathrm{d}P - \int f \mathrm{d}Q \right| = | \langle f , e \rangle_{\mathcal{H}(k)} | \leq \| f \|_{\mathcal{H}(k)} \| e \|_{\mathcal{H}(k)} ,
$$
which shows that
$$
0 \leq D_k(P,Q) = \sup_{\|f\|_{\mathcal{H}(k)} \leq 1} \left| \int f \mathrm{d}P - \int f \mathrm{d}Q \right| \leq \| e \|_{\mathcal{H}(k)} .
$$
If $\|e\|_{\mathcal{H}(k)} = 0$ then the bound is necessarily an equality.
If not, consider $f = e / \|e\|_{\mathcal{H}(k)}$, which satisfies $\|f\|_{\mathcal{H}(k)} \leq 1$ and $\langle f , e \rangle_{\mathcal{H}(k)} = \|e\|_{\mathcal{H}(k)}$, showing that the bound is in fact an equality.
\end{proof}

\noindent This result is summarised visually in \Cref{fig: kme}.

If $D_k(P,Q_n) = 0$, the cubature rule based on $Q_n$ will be exact for all integrands $f \in \mathcal{H}(k)$.
Does this mean that $Q_n$ and $P$ are identical?

\begin{definition}[Characteristic kernel]
A kernel $k$ is said to be \emph{characteristic} if $D_k(P,Q) = 0$ implies $P = Q$.
\end{definition}

\begin{example}[Polynomial kernel is not characteristic]
From \Cref{eq: finite dim rkhs}, the kernel $k(x,y) = \sum_{i=1}^p x^i y^i$ reproduces an \ac{rkhs} whose elements are the polynomials of degree at most $p$ on the domain $\mathcal{X} = \mathbb{R}$.
Thus $D_k(P,Q) = 0$ if and only if the moments $\int x^i \mathrm{d}P(x)$ and $\int x^i \mathrm{d}Q(x)$ are identical for $i = 1,\dots,p$.
In particular, $k$ is \emph{not} a characteristic kernel.
\end{example}

\begin{example}
The Gaussian kernel $k(\bm{x},\bm{y}) = \exp(-\|\bm{x}-\bm{y}\|^2)$ is a characteristic kernel on $\mathcal{X} = \mathbb{R}^d$.
\end{example}

The characteristic property is desirable but, on its own, it does not provide strong justification for using $D_k(P,Q)$ to measure the discrepancy between $P$ and $Q$.
For this reason we now introduce a stronger property, called \textit{weak convergence control}.
Let $Q_n \Rightarrow P$ denote that the sequence $(Q_n)_{n=1}^\infty$ converges \textit{weakly} (or \textit{in distribution}) to $P$ (i.e. $\int f \mathrm{d}Q_n \rightarrow \int f \mathrm{d}P$ for all functions $f$ which are continuous and bounded).

\begin{definition}[Weak convergence control]
A kernel $k$ is said to have \emph{weak convergence control} if $D_k(P,Q_n) \rightarrow 0$ implies that $Q_n \Rightarrow P$.
\end{definition}

\begin{remark}
Perhaps surprisingly, for a compact Hausdorff space $\mathcal{X}$, a bounded\footnote{and measurable} characteristic kernel $k$ is guaranteed to have weak convergence control.
This equivalence no longer holds when the domain $\mathcal{X}$ is non-compact, and a bounded and characteristic kernel can fail to have weak convergence control; see \citet{simon2020metrizing}.
Clearly a kernel that is not characteristic fails to have weak convergence control.
\end{remark}

\noindent Convergence control justifies attempting to minimise \ac{mmd} for the purposes of quantisation and more general approximation, as we will attempt in the sequel.

\begin{example}
The Gaussian kernel $k(\bm{x},\bm{y}) = \exp(-\|\bm{x}-\bm{y}\|^2)$ controls weak convergence of probability distributions on $\mathcal{X} = [0,1]^d$.
It can also be shown that the Gaussian kernel controls weak convergence on $\mathcal{X} = \mathbb{R}^d$; this can be deduced from e.g. Theorem 7 of \citet{simon2020metrizing} and the general results in \citet{sriperumbudur2011universality}.
\end{example}

\subsection{Optimal Quantisation} \label{subsec: opt quant mmd}

The goal of quantisation is to find $Q$ of the form $Q_n = \frac{1}{n} \sum_{i=1}^n \delta(\bm{x}_i)$ such that $Q_n \approx P$ in some sense, and in these lectures that sense will be MMD.
To start to move toward practical algorithms, notice that \Cref{lem: mmd ad embed} provides a means to compute \ac{mmd}:
\begin{align*}
D_k(P,Q)^2 & = \|\mu_P - \mu_Q\|_{\mathcal{H}(k)}^2 \\
& = \langle \mu_P - \mu_Q , \mu_P - \mu_Q \rangle_{\mathcal{H}(k)} \\
& = \langle \mu_P , \mu_P \rangle_{\mathcal{H}(k)} - 2 \langle \mu_P , \mu_Q \rangle_{\mathcal{H}(k)} + \langle \mu_Q , \mu_Q \rangle_{\mathcal{H}(k)} 
\end{align*}
Now, considering for example the term $\langle \mu_P , \mu_Q \rangle_{\mathcal{H}(k)}$, we have
\begin{align*}
\langle \mu_P , \mu_Q \rangle_{\mathcal{H}(k)} & = \left\langle \int k(\cdot,\bm{x}) \mathrm{d}P(\bm{x}) , \int k(\cdot,\bm{y}) \mathrm{d}Q(\bm{y}) \right\rangle_{\mathcal{H}(k)} \\
& = \iint \langle k(\cdot,\bm{x}) , k(\cdot,\bm{y}) \rangle_{\mathcal{H}(k)} \mathrm{d}P(\bm{x}) \mathrm{d}Q(\bm{y}) \\
& = \iint k(\bm{x},\bm{y}) \mathrm{d}P(\bm{x}) \mathrm{d}Q(\bm{y}).
\end{align*}
Here we have used the reproducing property, as well as using \Cref{lem: kme} to justify the exchanges of integral and inner product.
Proceeding similarly with all three terms results in the expression
\begin{align*} 
D_k(P,Q)^2 & = \iint k(\bm{x},\bm{y}) \mathrm{d}P(\bm{x}) \mathrm{d}P(\bm{y}) - 2 \iint k(\bm{x},\bm{y}) \mathrm{d}P(\bm{x}) \mathrm{d}Q(\bm{y}) + \iint k(\bm{x},\bm{y}) \mathrm{d}Q(\bm{x}) \mathrm{d}Q(\bm{y}) .
\end{align*}

As a simple baseline method for quantisation we consider Monte Carlo (\Cref{fig: mc}):

\begin{figure}[t!]
\centering
\includegraphics[width = 0.49\textwidth,clip,trim = 4cm 10.5cm 4cm 9cm]{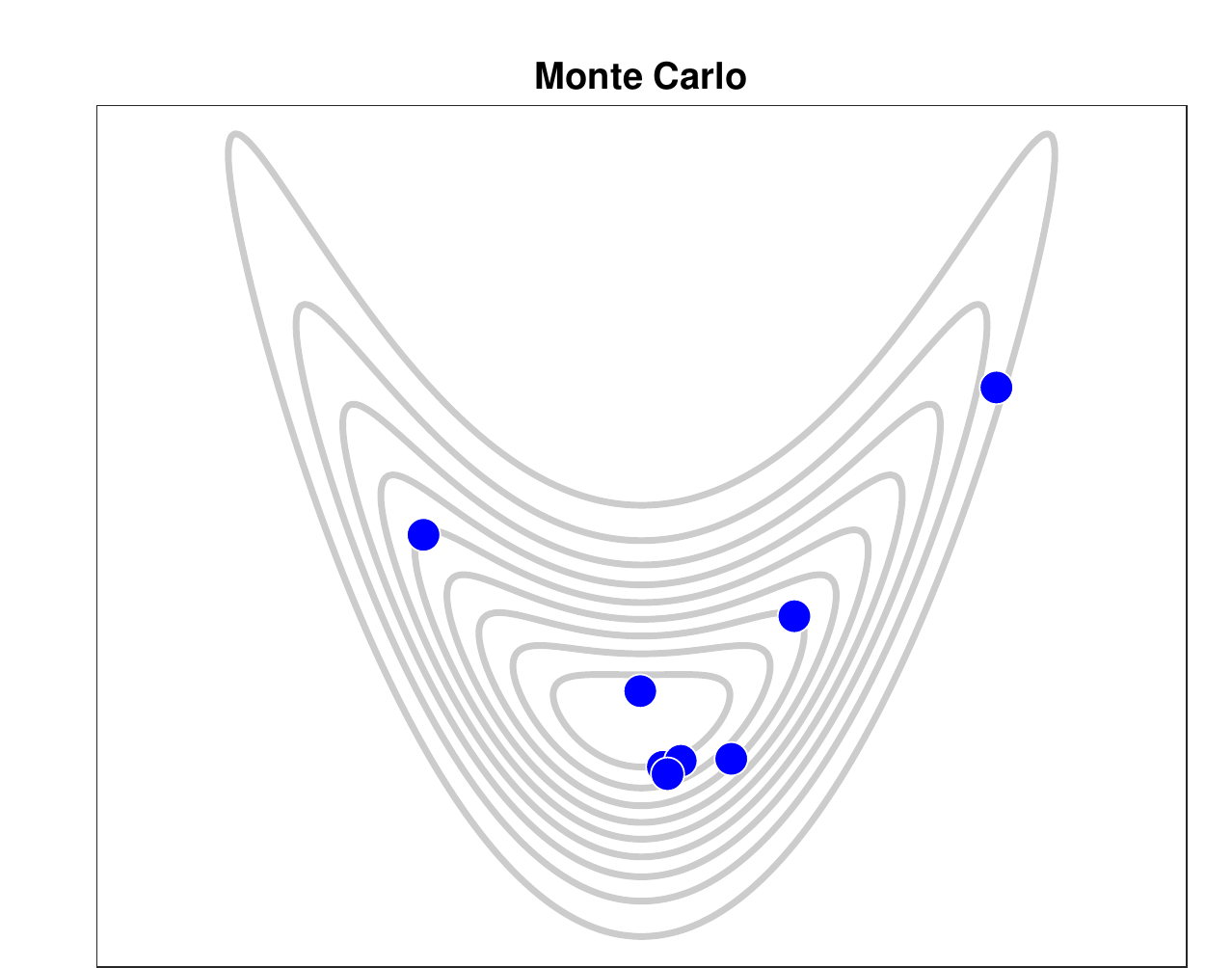}
\caption{\textit{Monte Carlo:} Independent samples (blue circles) from a ``horseshoe'' distribution $P$ (grey contours). }
\label{fig: mc}
\end{figure}

\begin{proposition}[MMD of Monte Carlo] \label{prop: MMD MC}
Let $\bm{x}_1,\dots,\bm{x}_n \sim P$ be independent.
Assume that $C := \int k(\bm{x},\bm{x}) \mathrm{d}P(\bm{x}) < \infty$.
Then
$$
\mathbb{E}\left[ D_k(P,Q_n)^2 \right] \leq \frac{C}{n} .
$$
\end{proposition}
\begin{proof}
From the above discussion, with $Q = Q_n = \frac{1}{n} \sum_{i=1}^n \delta(\bm{x}_i)$ we obtain that
\begin{align*}
D_k(P,Q_n)^2 = \iint k(\bm{x},\bm{y}) \mathrm{d}P(\bm{x}) \mathrm{d}P(\bm{y}) - \frac{2}{n} \sum_{i=1}^n \int k(\bm{x},\bm{x}_i) \mathrm{d}P(\bm{x}) + \frac{1}{n^2} \sum_{i=1}^n \sum_{j=1}^n k(\bm{x}_i,\bm{x}_j)
\end{align*}
Taking expectations of both sides gives
\begin{align*}
\mathbb{E}\left[ D_k(P,Q_n)^2 \right] & = \mathbb{E}\left[ \frac{1}{n^2} \sum_{i=1}^n \sum_{j=1}^n k(\bm{x}_i,\bm{x}_j) \right] -  \iint k(\bm{x},\bm{y}) \mathrm{d}P(\bm{x}) \mathrm{d}P(\bm{y}) \\
& = \mathbb{E}\left[ \frac{1}{n^2} \sum_{i=1}^n k(\bm{x}_i,\bm{x}_i) + \frac{1}{n^2} \sum_{i = 1}^n \sum_{j \neq i} k(\bm{x}_i,\bm{x}_j) \right] -  \iint k(\bm{x},\bm{y}) \mathrm{d}P(\bm{x}) \mathrm{d}P(\bm{y}) \\
& = \mathbb{E}\left[ \frac{1}{n^2} \sum_{i=1}^n k(\bm{x}_i,\bm{x}_i) \right] - \underbrace{ \frac{1}{n} \iint k(\bm{x},\bm{y}) \mathrm{d}P(\bm{x}) \mathrm{d}P(\bm{y}) }_{\geq 0} \\
& \leq \frac{1}{n} \int k(\bm{x},\bm{x}) \mathrm{d}P(\bm{x})
\end{align*}
since $\bm{x}_i \sim P$ are independent.
\end{proof}

Thus Monte Carlo sampling provides a consistent but potentially far from optimal quantisation of $P$.
Note that the convergence \textit{rate} in \Cref{prop: MMD MC} does not depend on the kernel $k$, which highlights the inefficiency of the Monte Carlo method in this context (e.g. compare against the later \Cref{thm: reweight}).
The goal of \emph{optimal} quantisation is to quantise $P$ using as few states $\bm{x}_i$ as possible (for a given approximation quality).
A conceptually simple approach to optimal quantisation is illustrated in \Cref{fig: own}, and you are encouraged to try this out:

\begin{exercise}[Optimal quantisation with MMD] \label{ex: Gaussian kernel cubature}
Consider $P = \mathcal{N}(\bm{0},\bm{I})$ and $k(\bm{x},\bm{y}) = \exp( - \|\bm{x} - \bm{y}\|^2 / \sigma^2)$ on $\mathcal{X} = \mathbb{R}^d$.
(You may wish to focus on $d=1$ or $d=2$.)
\begin{enumerate}[label=(\alph*)]
\item Calculate (analytically) the kernel mean embedding $\mu_P(\bm{x}) = \int k(\bm{x},\bm{y}) \mathrm{d}P(\bm{y})$.
\item For a fixed value of $n$ (e.g. $n=10$) and a fixed value of $\sigma$ (e.g. $\sigma = 1$), try to numerically optimise the locations of the states $\bm{x}_1,\dots,\bm{x}_n$ in order to minimise $D_k(P,\frac{1}{n} \sum_{i=1}^n \delta(\bm{x}_i) )$.
\item What effect does varying the bandwidth parameter $\sigma$ have on the approximations that are produced?
\end{enumerate}
\end{exercise}

\begin{figure}[t!]
\centering
\includegraphics[width = 0.49\textwidth,clip,trim = 4cm 10.5cm 4cm 9cm]{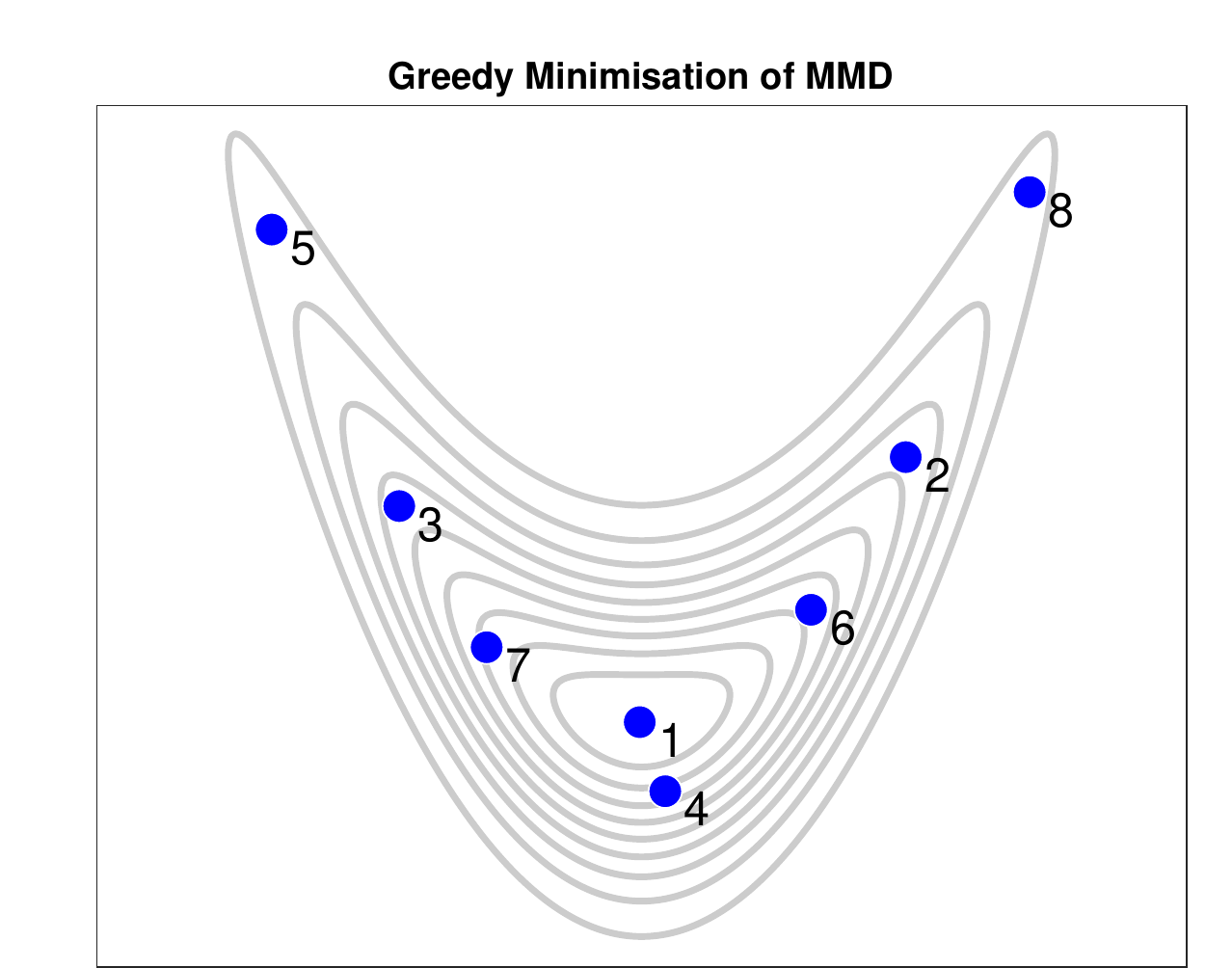}
\caption{\textit{Optimal (un-weighted) quantisation:} Sequential (greedy) minimisation of \ac{mmd} to select $n=8$ states $\{\bm{x}_i\}_{i=1}^n$ in an approximation $Q_n$ to $P$.  The numbers indicate the order in which the states $\bm{x}_i$ were selected.}
\label{fig: own}
\end{figure}

After performing \Cref{ex: Gaussian kernel cubature}, it should be clear that (1) \ac{mmd} provides a coherent framework for optimal quantisation, but (2) direct/naive numerical optimisation of \ac{mmd} may be impractical, or at least not straight forward.
A neat solution is proposed next.

\subsection{Optimal Approximation} 

Motivated by the difficulty of multivariate optimisation in $\mathcal{X} \times \dots \times \mathcal{X}$ in \Cref{ex: Gaussian kernel cubature}, in this section we return to the case of independently sampled $\bm{x}_i \sim P$ but now we allow for \textit{weighted} point sets; i.e. approximations of the form $Q_n = \sum_{i=1}^n w_i \delta(\bm{x}_i)$ for some weights $w_1,\dots,w_n \in \mathbb{R}$.

\begin{lemma} \label{lem: opt weights}
Let $\bm{x}_1,\dots,\bm{x}_n \in \mathcal{X}$ be distinct.
The optimal weights 
\begin{align*}
\argmin_{\bm{w} \in \mathbb{R}^n} D_k\left(P , \sum_{i=1}^n w_i \delta(\bm{x}_i) \right)
\end{align*}
are the solution of the linear system 
\begin{align}
\bm{K} \bm{w} & = \bm{z}  \label{eq: lin sys}
\end{align}
where $K_{ij} = k(\bm{x}_i,\bm{x}_j)$ and $z_i = \mu_P(\bm{x}_i)$.
\end{lemma}
\begin{proof}
The \ac{mmd} between $P$ and $Q_n = \sum_{i=1}^n w_i \delta(\bm{x}_i)$ can be expressed as
\begin{align*}
D_k(P,Q_n)^2 & = \iint k(\bm{x},\bm{y}) \mathrm{d}P(\bm{x}) \mathrm{d}P(\bm{y}) - 2 \iint k(\bm{x},\bm{y}) \mathrm{d}P(\bm{x}) \mathrm{d}Q(\bm{y}) + \iint k(\bm{x},\bm{y}) \mathrm{d}Q(\bm{x}) \mathrm{d}Q(\bm{y}) \\
& = C - 2 \bm{z}^\top \bm{w} + \bm{w}^\top \bm{K} \bm{w}
\end{align*}
where $C = \iint k(\bm{x},\bm{y}) \mathrm{d}P(\bm{x}) \mathrm{d}P(\bm{y})$ is independent of $\bm{w}$.
This is a non-degenerate quadratic form in $\bm{w}$ (since $\bm{K}$ is a positive definite matrix), from which the result is easily verified (e.g. differentiate w.r.t. $\bm{w}$ and solve for the unique critical point).
\end{proof}

\noindent See the illustration in \Cref{fig: omc}.

\begin{figure}[t!]
\centering
\includegraphics[width = 0.49\textwidth,clip,trim = 4cm 10.5cm 4cm 9cm]{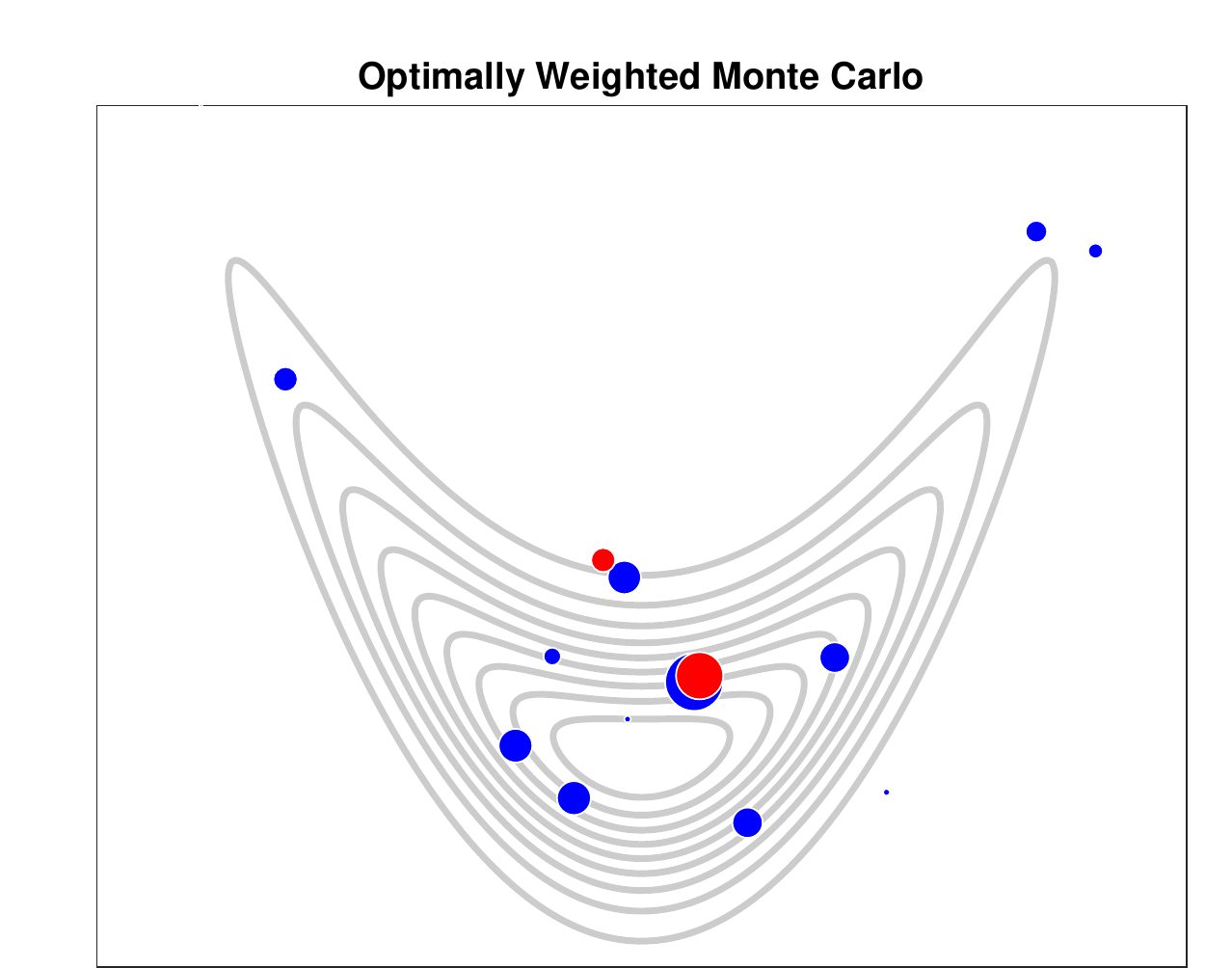}
\caption{\textit{Optimally weighted Monte Carlo samples:} The weights $w_1,\dots,w_n$ are obtained by minimising \ac{mmd} in the manner of \Cref{lem: opt weights}.
Blue indicates states $\bm{x}_i$ with positive weights $w_i>0$, while red indicates negative weights $w_i < 0$.
The size of the circles is proportional to $|w_i|$.  }
\label{fig: omc}
\end{figure}

\begin{remark}
The linear system in \eqref{eq: lin sys}, defining the optimal weights, can be numerically ill-conditioned (i.e. close to singular) when $n$ is large or when two of the $\bm{x}_i$ are close together (where the terms ``large'' and ``close'' depend on the specific kernel $k$ that is used).
(Of course, if $\bm{x}_i$ and $\bm{x}_j $ are identical we may assign $w_i = 0$ without loss of generality.)
Although we will not discuss this point further, there are several techniques for numerical regularisation could be used.
\end{remark}

\begin{remark}
In general the optimal weights can be negative and will not sum to $1$.
Both constraints can be enforced if required and, essentially, the results that we discuss continue to hold.
See e.g. Section 2.3 of \citet{karvonen2018bayes}.
\end{remark}

Our aim in the remainder of this section is to show that the optimal weights in \Cref{lem: opt weights} enable faster rates of convergence than the Monte Carlo rate in \Cref{prop: MMD MC}.
To make this concrete we consider a particularly natural class of \ac{rkhs}, introduced next.
The \textit{multi-index} notation
$$
\partial^{\bm{\alpha}}f : \bm{x} \mapsto \frac{\partial^{\bm{\alpha}}f(\bm{x})}{\partial \bm{x}^{\bm{\alpha}}} = \frac{\partial^{\|\bm{\alpha}\|_1}}{ \partial x_1^{\alpha_1} \dots \partial x_d^{\alpha_d}} f(\bm{x}), \qquad \bm{\alpha} \in \mathbb{N}_0^d
$$ 
will be used, and we let $|\bm{\alpha}| = \alpha_1 + \dots + \alpha_d$.

\begin{definition}
For $s > d/2$ and (sufficiently regular) $\mathcal{X} \subset \mathbb{R}^d$, the (order $s$) Sobolev space $H^s(\mathcal{X})$ is defined to be the set of functions $f : \mathcal{X} \rightarrow \mathbb{R}$ whose mixed partial derivatives $\partial^{\bm{\alpha}} f$, $|\bm{\alpha}| \leq s$, exist in $L^2(\mathcal{X})$.
This becomes a Hilbert space with inner product
$$
\langle f , g \rangle_{H^s(\mathcal{X})} = \sum_{|\bm{\alpha}| \leq s} \int \frac{\partial^{\bm{\alpha}}f(\bm{x})}{\partial \bm{x}^{\bm{\alpha}}}  \frac{\partial^{\bm{\alpha}}g(\bm{x})}{\partial \bm{x}^{\bm{\alpha}}}  \mathrm{d}\bm{x}  .
$$
A kernel $k : \mathcal{X} \times \mathcal{X} \rightarrow \mathbb{R}$ is called a \emph{Sobolev kernel} if there exists $0 < c_1 < c_2 < \infty$ such that, for all $f \in \mathcal{H}(k)$, we have $c_1 \|f\|_{H^s(\mathcal{X})} \leq \|f\|_{\mathcal{H}(k)} \leq c_2 \|f\|_{H^s(\mathcal{X})}$.
\end{definition}

\begin{example} \label{ex: Sobolev kernels}
Let $z_+^m$ denote $\max(0,z)^m$ in shorthand.
Then examples of Sobolev kernels on (sufficiently regular) $\mathcal{X} \subseteq \mathbb{R}$ include the following, due to \citet{wendland1998error}:
\begin{center}
\begin{tabular}{|l|l|} \hline
$k(x,y)$ \; \; {\normalfont ($r = |x-y|$, $x,y \in \mathbb{R}$)} & {\normalfont order} \\ \hline
$(1 - r )_+$ & $s=1$ \\
$(1 - r)_+^3 (3r+1)$  & $s=2$  \\
$(1 - r)_+^5 (8r^2 + 5r + 1)$ &  $s=3$ \\ \hline
\end{tabular}
\end{center}
These kernels are convenient for numerical reasons, due to their compact support, which renders $\bm{K}$ a \emph{sparse} matrix.
\end{example}

The above inner product should be contrasted with that in \Cref{lem: Hilbert space}. Here we are considering mixed partial derivatives of order at most $s$, where each coordinate can in principle be differentiated more than once, whereas in \Cref{lem: Hilbert space} we consider mixed partial derivatives of up to order $d$, provided that each coordinate of $\bm{x}$ is differentiated at most once.


Given the slow convergence of \ac{mmd} with uniformly-weighted random samples in \Cref{prop: MMD MC}, the following is possibly very surprising:

\begin{theorem} \label{thm: reweight}
Let $\bm{x}_1,\dots,\bm{x}_n \sim P$ be independent and let $\bm{w} = \bm{w}(\bm{x}_1,\dots,\bm{x}_n)$ denote optimal weights in the sense of \Cref{lem: opt weights}.
Let $k(\bm{x},\bm{y})$ be an (order $s$) Sobolev kernel.
Then, under regularity conditions on the domain $\mathcal{X}$, which is of dimension $d$, and the distribution $P$, there exists a constant $0 < C < \infty$ such that
$$
\mathbb{E}\left[ D_k\left(P , \sum_{i=1}^n w_i \delta(\bm{x}_i) \right) \right] \leq \left(\frac{C \log(n)}{n}\right)^{s/d} .
$$
\end{theorem}
\begin{proof}
The following is a sketch:
For simplicity assume that $\mathcal{H}(k)$ and $H^s(\mathcal{X})$ have an identical inner product.
Under regularity conditions on the domain $\mathcal{X}$, which amounts to $\mathcal{X}$ being bounded and satisfying an \emph{interior cone condition}, one can obtain a \emph{sampling inequality} of the form
$$
\|f - f_n\|_\infty \leq C h_n^{s/d} \|f\|_{H^s(\mathcal{X})}
$$
where $f_n(\bm{x}) = \sum_{i=1}^n c_i k(\bm{x},\bm{x}_i)$, $\bm{c} = \bm{K}^{-1} \bm{f}$, is 
an interpolant of the function $f$ at the locations $\bm{x}_1,\dots,\bm{x}_n$, and $h_n$ is the fill distance
$$
h_n = \sup_{\bm{x} \in \mathcal{X}} \; \min_{i = 1,\dots,n} \|\bm{x} - \bm{x}_i \| ,
$$
see e.g. Theorem 11.13 of \citet{wendland2004scattered}.
Then observe that $\int f_n \mathrm{d}P = \sum_{i=1}^n c_i z_i = \bm{f}^\top \bm{K}^{-1} \bm{z} = \bm{f}^\top \bm{w} = \sum_{i=1}^n w_i f(\bm{x}_i)$.
Thus $\left| \sum_{i=1}^n w_i f(\bm{x}_i) - \int f \mathrm{d}P \right| = \left| \int f_n \mathrm{d}P - \int f \mathrm{d}P \right| \leq \|f_n - f\|_\infty \leq C h_n^{s/d} \|f\|_{H^s(\mathcal{X})}$.
From the definition of \ac{mmd} it follows that $D_k(P,\sum_{i=1}^n w_i f(\bm{x}_i)) \leq C h_n^{s/d}$.
Under regularity conditions on $P$, which amounts to $P$ admitting a \ac{pdf} bounded away from 0 and $\infty$ on $\mathcal{X}$, one can show that $\mathbb{E}[h_n^{s/d}]$ decreases at the advertised $\mathcal{O}(((\log n) / n)^{s/d})$ rate; see \citet{reznikov2016covering}.
\end{proof}

\noindent Thus for $s = d/2$ we recover the same rate as \Cref{prop: MMD MC} for un-weighted Monte Carlo (up to log factors), while for $s > d/2$ we obtain faster convergence in \ac{mmd}. 

\begin{remark}
Stronger concentration inequalities than \Cref{thm: reweight} can also be established; see e.g. \citet{ehler2019optimal}.
\end{remark}

So surely it is a good idea to employ optimal weights?
Not necessarily - the computational cost is $\mathcal{O}(n^3)$ in general and numerical ill-conditioning requires careful treatment.
In \ac{qmc} an important goal is to find a deterministic sequence of sets of (un-weighted) states whose computation is $\mathcal{O}(n)$, such that \ac{mmd} is asymptotically minimised.
Thus, at least for the simple forms of $P$ for which a \ac{qmc} method has been discovered, the \ac{qmc} approach would usually be preferred.
The \ac{rkhs}/\ac{mmd} perspective on \ac{qmc} was popularised by \citet{hickernell1998generalized}.

\begin{exercise}[Rates of convergence and MMD] \label{exer: rates mmd}
Consider the uniform distribution $P = \mathcal{U}([0,1])$ together with the Sobolev kernels $k$ of orders $s \in \{1,2,3\}$ in \Cref{ex: Sobolev kernels}.
\begin{enumerate}[label=(\alph*)]
\item Calculate the kernel mean embeddings $\mu_P(x) = \int_0^1 k(x,y) \mathrm{d}y$.
\item Generate and store a sequence $(x_i)_{i=1}^{100}$ of independent samples from $P$.
\item For each $n = 1,\dots,100$ and each $s = 1,2,3$, calculate and plot the values of
$$
D_k\left(P, \frac{1}{n} \sum_{i=1}^n \delta(x_i) \right)
\qquad 
\text{and}
\qquad
D_k\left(P, \sum_{i=1}^n w_i^{(n)} \delta(x_i) \right)
$$
where $\bm{w}^{(m)} = (w_1^{(m)},\dots,w_n^{(m)})$ are the optimal weights obtained by solving the $m$-dimensional linear system in \Cref{lem: opt weights}, based on the states $\{x_i\}_{i=1}^m$.
\item What rates of convergence would you expect to observe for these quantities as $n \rightarrow \infty$, and do your experiments agree with these rates of convergence?
\end{enumerate}
\end{exercise}

\paragraph*{Chapter Notes}

The presentation of \Cref{lem: kme} follows Lemma 3.1 of  \citet{muandet2016kernel}.
\Cref{lem: kme} presented an elementary argument for why kernel mean embeddings are well-defined, but a more general framework is the Bochner integral.
Bochner's criterion for integrability states that a Bochner-measurable function $F : [0,1]^d \rightarrow \mathcal{H}(k)$ is Bochner integrable if and only if $\int \|F(\bm{x})\|_{\mathcal{H}(k)} \mathrm{d}P(\bm{x}) < \infty$.
Here we take $F(\bm{x}) = k(\cdot,\bm{x})$, noting that $\|k(\cdot,\bm{x})\|_{\mathcal{H}(k)} = \sqrt{k(\bm{x},\bm{x})}$, which recovers the condition in \Cref{lem: kme}.
There is an elegant criterion to determine when a translation-invariant kernel $k$ (i.e. $k(\bm{x},\bm{y}) = \phi(\bm{x} - \bm{y})$ for some function $\phi$) is characteristic; see Section 2.1 of \citet{muandet2016kernel}.
Stronger concentration inequalities than \Cref{prop: MMD MC} for Monte Carlo MMD have been established; see Section 3.3 of \citet{muandet2016kernel}.
The case of \Cref{thm: reweight} where $\mathcal{X}$ is a smooth, connected, closed Riemannian manifold of dimension $d$ is presented in \citet{ehler2019optimal}.
(This requires some generalisation of the definition of a Sobolev kernel.)
The fastest known rates for explicit constructions of weighted approximations in the case of \Cref{ex: Gaussian kernel cubature} are (at the time of writing) due to \citet{karvonen2021integration}.
Greedy optimisation can provide a practical solution to optimal quantisation problems like \Cref{ex: Gaussian kernel cubature} \citep{pronzato2020bayesian,teymur2021optimal}, as can the sophisticated numerical methods for high-dimensional optimisation used to produce \Cref{fig: Graf} \citep{graf2012quadrature}.
Several tricks are available to reduce the cost of solving the linear system in \eqref{eq: lin sys}; see e.g. \citet{karvonen2018fully}.

\section{Stein Discrepancy}

In this final section we aim to perform optimal quantisation of a distribution $P$ that admits a \ac{pdf} $p(\bm{x})$ on $\bm{x} \in \mathbb{R}^d$, such that
$$
p(\bm{x}) = \frac{\tilde{p}(\bm{x})}{Z} ,
$$
where $\tilde{p}$ can be exactly evaluated but $Z$, and hence $p(\bm{x})$, cannot easily be evaluated or even approximated.
This setting is typical in applications of Bayesian inference, where we have
$$
p(\bm{x}) = \frac{\pi(\bm{x}) \mathcal{L}(\bm{x})}{Z}
$$
where $\pi(\bm{x})$ is a \emph{prior} \ac{pdf}, $\mathcal{L}(\bm{x})$ is a likelihood, and the implicitly defined normalisation constant $Z$ is the \emph{marginal likelihood}.
The integral
$$
Z = \int \pi(\bm{x}) \mathcal{L}(\bm{x}) \mathrm{d}\bm{x}
$$
is often extremely challenging to evaluate due to localised regions in which $\mathcal{L}$ takes very large values.
Several methods have been developed in the statistics, applied probability, physics and machine learning literatures to approximate distributions $P$ with these characteristics, including \textit{\ac{mcmc}}, \textit{\ac{smc}}, and \textit{variational inference}.
These techniques do not typically attempt \textit{optimal} quantisation, since even the basic quantisation task can be difficult.

The aim of this section is to discuss whether the techniques described in \Cref{sec: mmd} can be applied in this more challenging context. 
The apparent difficulty is that we cannot compute integrals with respect to $P$, such as $\int k(\cdot,\bm{x}) \mathrm{d}P(\bm{x})$, which are required for computation of \ac{mmd}.
A hint at a possible solution is provided by the following result:

\begin{lemma} \label{lem: Stein discrep}
Suppose $k_P : \mathcal{X} \times \mathcal{X} \rightarrow \mathbb{R}$ is a symmetric positive definite kernel with $\int k_P(\cdot,\bm{x}) \mathrm{d}P = 0$ for all $\bm{x} \in \mathcal{X}$.
Then
$$
D_{k_P}(Q) = D_{k_P}(P,Q) = \sup_{\|f\|_{\mathcal{H}(k_P)} \leq 1} \left| \int f \mathrm{d}Q \right| .
$$
\end{lemma}
\begin{proof}
For all $f \in \mathcal{H}(k_P)$ it holds that $\int f \mathrm{d}P = 0$, whence the result.
Indeed, from the reproducing property, and using \Cref{lem: kme} with \Cref{ass: trace class} to interchange integral with inner product, $\int f \mathrm{d}P = \int \langle f , k(\cdot,\bm{x}) \rangle_{\mathcal{H}(k_P)} \mathrm{d}P(\bm{x}) = \left\langle f , \int k(\cdot,\bm{x}) \mathrm{d}P(\bm{x}) \right\rangle_{\mathcal{H}(k_P)} = \langle f , 0 \rangle_{\mathcal{H}(k_P)} = 0$.
\end{proof}

\noindent The important point here is that $D_{k_P}(Q)$ does not require integrals with respect to $P$ to be computed.
A kernel $k_P$ with $\int k_P(\cdot,\bm{x}) \mathrm{d}P = 0$ will be called a \emph{Stein kernel} (for $P$), for reasons that will become clear in the sequel.
An example for how such a kernel can be constructed is as follows:
Consider the bounded linear operator $(\mathcal{A}_P g)(\bm{x}) = g(\bm{x}) - \int g \mathrm{d}P$ acting on elements of an \ac{rkhs} $\mathcal{H}(k)$.
If we apply $\mathcal{A}_P$ to both arguments of the kernel $k$, we obtain a Stein kernel
\begin{align}
k_P(\bm{x},\bm{y}) & = \mathcal{A}_P^{\bm{y}} \mathcal{A}_P^{\bm{x}} k(\bm{x},\bm{y}) \nonumber \\
& = k(\bm{x},\bm{y}) - \int k(\bm{x},\bm{y}) \mathrm{d}P(\bm{x}) - \int k(\bm{x},\bm{y}) \mathrm{d}P(\bm{y}) + \iint k(\bm{x},\bm{y}) \mathrm{d}P(\bm{x}) \mathrm{d}P(\bm{y}). \label{eq: zero mean kernel}
\end{align}
Indeed, $\int k_P(\cdot,\bm{x}) \mathrm{d}P(\bm{x}) = \int \mathcal{A}_P^{\bm{y}} \mathcal{A}_P^{\bm{x}} k(\bm{x},\bm{y}) \mathrm{d}P(\bm{x}) = \mathcal{A}_P^{\bm{y}} \int \mathcal{A}_P^{\bm{x}} k(\bm{x},\bm{y}) \mathrm{d}P(\bm{x}) = \mathcal{A}_P^{\bm{y}} 0 = 0$, where interchange of $\mathcal{A}_P^{\bm{y}}$ and the integral is justified by noting that $\mathcal{A}_P^{\bm{y}}$ is a bounded linear operator and following similar reasoning to \Cref{lem: kme}.
In fact, the \ac{rkhs} $\mathcal{H}(k_P)$ consists of functions of the form $\mathcal{A}_P g = g - \int g \mathrm{d}P$ where $g \in \mathcal{H}(k)$.
Unfortunately, the Stein kernel in \eqref{eq: zero mean kernel} is not useful because it still involves the problematic integral $\int k(\cdot,\bm{x}) \mathrm{d}P(\bm{x})$.
The next section presents a more useful construction of a Stein kernel.

\subsection{Stein Operators}

The aim here is to identify an alternative operator $\mathcal{A}_P$, which \textit{can} be computed.
Let $\nabla f = (\partial_{x_1}f,\dots,\partial_{x_d}f)^\top$ for differentiable functions $f : \mathbb{R}^d \rightarrow \mathbb{R}$.
Our main tool is a \emph{Stein operator}, and the classical example of this is as follows:

\begin{assumption} \label{ass: lipz}
The distribution $P$ admits a positive and differentiable \ac{pdf} such that $\bm{x} \mapsto (\nabla \log p)(\bm{x})$ is Lipschitz.
\end{assumption}

\begin{definition}[Canonical Stein operator] \label{def: can so}
For a distribution $P$ admitting a positive and differentiable density $p$ on $\mathbb{R}^d$, we define the \emph{canonical Stein operator}
$$
(\mathcal{A}_P g)(\bm{x}) = (\nabla \cdot g)(\bm{x}) + g(\bm{x}) \cdot (\nabla \log p)(\bm{x})
$$
acting on differentiable vector field $g : \mathbb{R}^d \rightarrow \mathbb{R}^d$, where $\bm{x} \in \mathbb{R}^d$.
\end{definition}

\noindent The canonical Stein operator was introduced (for Gaussian $P$) in \citet{stein1972bound}.
Importantly, observe that
$$
(\nabla \log p)(\bm{x}) = \frac{(\nabla p)(\bm{x})}{p(\bm{x})} = \frac{\frac{1}{Z} (\nabla \tilde{p})(\bm{x})}{\frac{1}{Z} \tilde{p}(\bm{x})} = \frac{(\nabla \tilde{p})(\bm{x})}{\tilde{p}(\bm{x})} = (\nabla \log \tilde{p})(\bm{x}) ,
$$
which can be computed without knowledge of $p$ or $Z$, provided $\tilde{p}$ and $\nabla \tilde{p}$ can be evaluated.
Loosely speaking, we can apply the Stein operator $\mathcal{A}_P$ in \Cref{def: can so} to a standard kernel $k$ to obtain the following Stein kernel:

\begin{lemma} \label{lem: Stein kernel}
Suppose that $k : \mathbb{R}^d \times \mathbb{R}^d \rightarrow \mathbb{R}$ is a symmetric positive definite kernel with 
$(\bm{x},\bm{y}) \mapsto \partial^{(\bm{\alpha},\bm{\beta})} k(\bm{x},\bm{y})$ being continuous and uniformly bounded for all $|\bm{\alpha}|, |\bm{\beta}| \leq 1$.
Suppose $\int \|\nabla \log p(\bm{x})\| \mathrm{d}P(\bm{x}) < \infty$ and that $\sup_{\|\bm{x}\| \geq r} r^{d-1} p(\bm{x}) \rightarrow 0$ as $r \rightarrow \infty$.
Then
\begin{align*}
k_P(\bm{x},\bm{y}) & = \nabla_{\bm{x}} \cdot \nabla_{\bm{y}} k(\bm{x},\bm{y}) + \nabla_{\bm{x}} \log p(\bm{x}) \cdot \nabla_{\bm{y}} k(\bm{x},\bm{y}) + \nabla_{\bm{y}} \log p(\bm{y}) \cdot \nabla_{\bm{x}} k(\bm{x},\bm{y}) \\
& \qquad + (\nabla_{\bm{x}} \log p(\bm{x})) \cdot (\nabla_{\bm{y}} \log p(\bm{y})) k(\bm{x},\bm{y})
\end{align*}
is a symmetric positive definite kernel with $\int k_P(\bm{x},\bm{y}) \mathrm{d}P(\bm{y}) = 0$ for all $\bm{x} \in \mathbb{R}^d$.
\end{lemma}
\begin{proof}
First notice that
\begin{align*}
k_P(\bm{x},\bm{y}) = \mathcal{A}_P^{\bm{y}} \left[ \begin{array}{c} \vdots \\  \nabla_{x_i} k(\bm{x},\bm{y}) + k(\bm{x},\bm{y}) \nabla_{x_i} \log p(\bm{x})  \\ \vdots \end{array} \right] = \mathcal{A}_P^{\bm{y}} g(\bm{y})
\end{align*}
where, under our assumptions, (a) $\bm{y} \mapsto g(\bm{y})$ is bounded, and (b) $\bm{y} \mapsto \nabla_{\bm{y}} \cdot g(\bm{y})$ is integrable with respect to $P$.
Thus it suffices to show that $\int \mathcal{A}_P g \mathrm{d}P = 0$ for \emph{all} vector fields $g$ for which (a) and (b) hold.

Let $g$ be such a vector field, and let $B_r = \{\bm{x} \in \mathbb{R}^d : \|\bm{x}\| \leq r\}$ and $S_r = \{\bm{x} \in \mathbb{R}^d : \|\bm{x}\| = r\}$.
The main idea is to apply the divergence theorem (i.e. integrate by parts):
\begin{align*}
\int \mathcal{A}_P g \mathrm{d}P & = \int (\nabla \cdot g) + g \cdot (\nabla \log p) \mathrm{d}P \\
& = \int (\nabla \cdot (p g))(\bm{x}) \mathrm{d}\bm{x} \\
& = \lim_{r \rightarrow \infty} \int_{B_r} (\nabla \cdot (pg))(\bm{x}) \mathrm{d}\bm{x} \\
& = \lim_{r \rightarrow \infty} \oint_{S_r} p(\bm{x}) (g(\bm{x}) \cdot n(\bm{x})) \mathrm{d}\bm{x}
\end{align*}
where $n(\bm{x})$ is the outward unit normal to $S_r$ at $\bm{x}$.
(The regularity assumptions ensure that the integrals $\int (\nabla \cdot g) \mathrm{d}P$ and $\int g \cdot (\nabla \log p) \mathrm{d}P$ exist.)
Now
\begin{align*}
\oint_{S_r} p(\bm{x}) (g(\bm{x}) \cdot n(\bm{x})) \mathrm{d}\bm{x} & \leq \|g\|_\infty \sup_{\|\bm{x}\| \geq r} p(\bm{x}) \oint_{S_r} \mathrm{d}\bm{x} \\
& = \|g\|_\infty \sup_{\|\bm{x}\| \geq r} p(\bm{x}) \frac{2 \pi^{d/2}}{\Gamma(d/2)} r^{d-1} \\
& \rightarrow 0 \text{ as } r \rightarrow \infty ,
\end{align*}
where we have used the formula for the surface area of the radius $r$ sphere in $\mathbb{R}^d$.
\end{proof}

\begin{definition}[Kernel Stein discrepancy]
With $k_P$ defined in \Cref{lem: Stein kernel}, we call $D_{k_P}$ in \Cref{lem: Stein discrep} a \ac{ksd}.
\end{definition}

\begin{lemma} \label{lem: ksd explicit}
For $Q_n = \sum_{i=1}^n w_i \delta(\bm{x}_i)$, we have the explicit form of \ac{ksd}:
$$
D_{k_P}(Q_n) = \sqrt{ \sum_{i=1}^n \sum_{j=1}^n w_i w_j k_P(\bm{x}_i,\bm{x}_j) }
$$
\end{lemma}
\begin{proof}
Immediate from the closed form expression for \ac{mmd} in \Cref{subsec: opt quant mmd}, with $k_P$ in place of $k$ and using the fact that $\int k_P(\bm{x},\bm{y}) \mathrm{d}P(\bm{y}) = 0$ for all $\bm{x} \in \mathbb{R}^d$ from \Cref{lem: Stein kernel}.
\end{proof}

As with \ac{mmd}, we can establish properties analogous to characteristicness and convergence control for \ac{ksd}.
Here we focus on the stronger property of convergence control:

\begin{theorem} \label{thm: ksd cc}
Let $P$ be \emph{distantly dissipative}, meaning that $\lim \inf_{r \rightarrow \infty} \kappa(r) > 0$ where
$$
\kappa(r) = \inf\left\{ -2 \frac{(\langle \nabla \log p)(\bm{x}) - (\nabla \log p)(\bm{y}) , \bm{x} - \bm{y} \rangle}{\|\bm{x} - \bm{y}\|^2} : \|\bm{x} - \bm{y}\| = r \right\} .
$$
Consider the kernel $k(\bm{x},\bm{y}) = (\sigma^2 + \|\bm{x} - \bm{y}\|^2)^{-\beta}$ for some fixed $\sigma > 0$ and a fixed exponent $\beta \in (0,1)$.
Then $D_{k_P}(Q_n) \rightarrow 0$ implies $Q_n \Rightarrow P$.
\end{theorem}
\begin{proof}
Theorem 8 in \citet{gorham2017measuring}.
\end{proof}

\noindent \Cref{thm: ksd cc} justifies attempting to minimise \ac{ksd} from the point of view of quantisation, which we will discuss next.

\subsection{Optimal Quantisation}

The simplest use of \ac{ksd} for quantisation is as follows:

\begin{exercise}[Optimal quantisation with KSD] \label{ex: Stein points}
Consider $P = \mathcal{N}(\bm{0},\bm{I})$ and $k(\bm{x},\bm{y}) = (\sigma^2 + \|\bm{x} - \bm{y}\|^2)^{-1/2}$ on $\mathcal{X} = \mathbb{R}^d$.
(You may wish to focus on $d=1$ or $d=2$.)
\begin{enumerate}[label=(\alph*)]
\item Verify that $P$ and $k$ satisfy the conditions of \Cref{thm: ksd cc}.
\item Calculate (analytically) the Stein kernel $k_P(\bm{x},\bm{y})$.
\item For a fixed value of $n$ (e.g. $n=10$) and a fixed value of $\sigma$ (e.g. $\sigma = 1$), try to numerically optimise the locations of the states $\bm{x}_1,\dots,\bm{x}_n$ in order to minimise $D_k(P,\frac{1}{n} \sum_{i=1}^n \delta(\bm{x}_i) )$.
\item What effect does varying the bandwidth parameter $\sigma$ have on the approximations that are produced?
\item What happens if instead the Gaussian kernel $k(\bm{x},\bm{y}) = \exp(- \|\bm{x}-\bm{y}\|^2 / \sigma^2)$ is used?
\end{enumerate}
\end{exercise}

\noindent This provides an optimisation-based alternative to popular sampling-based algorithms, such as \ac{mcmc} and \ac{smc}.
See \Cref{fig: sp}.

\begin{figure}
\centering
\includegraphics[width = 0.49\textwidth,clip,trim = 4cm 10.5cm 4cm 9cm]{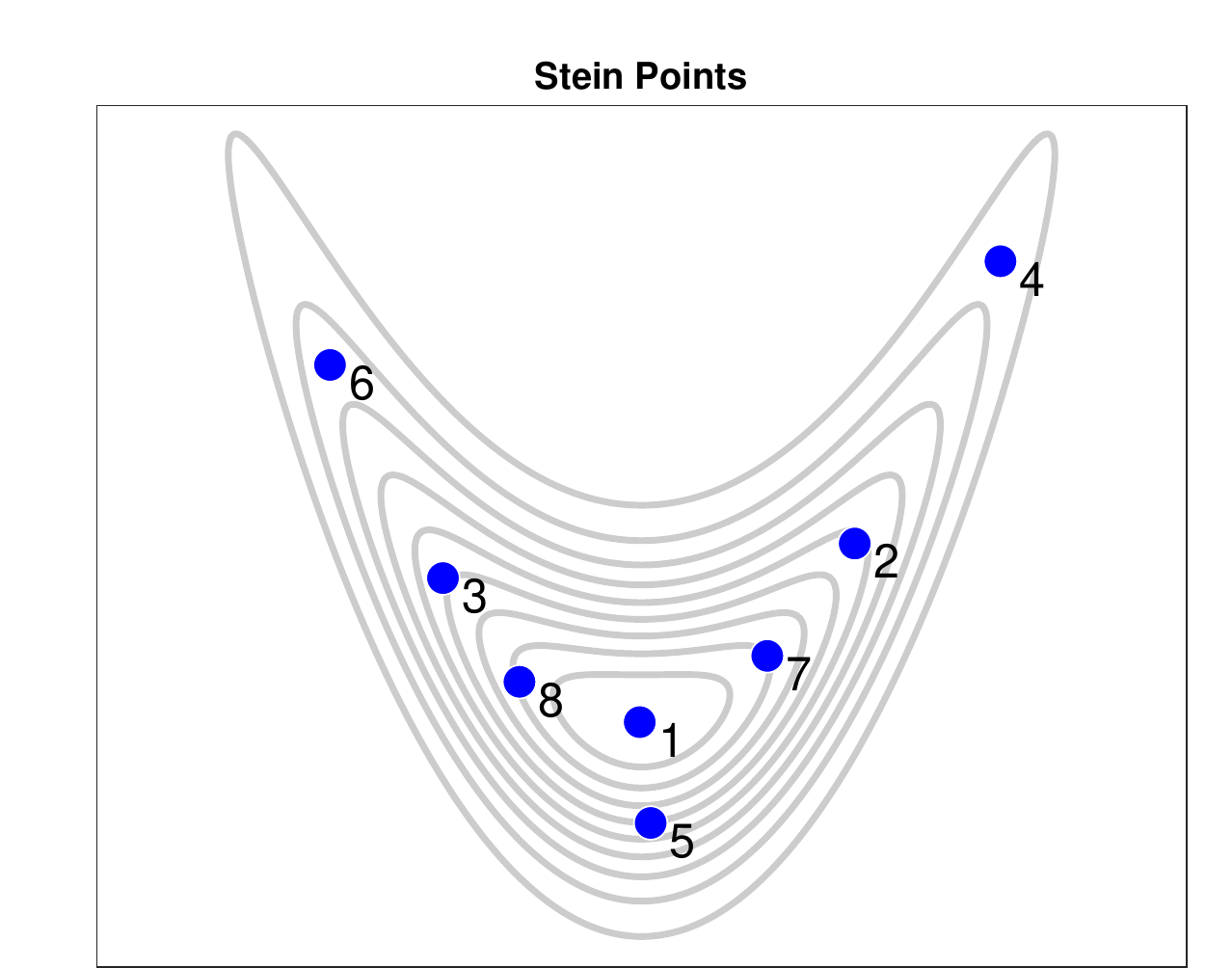}
\includegraphics[width = 0.49\textwidth,clip,trim = 4cm 10.5cm 4cm 9cm]{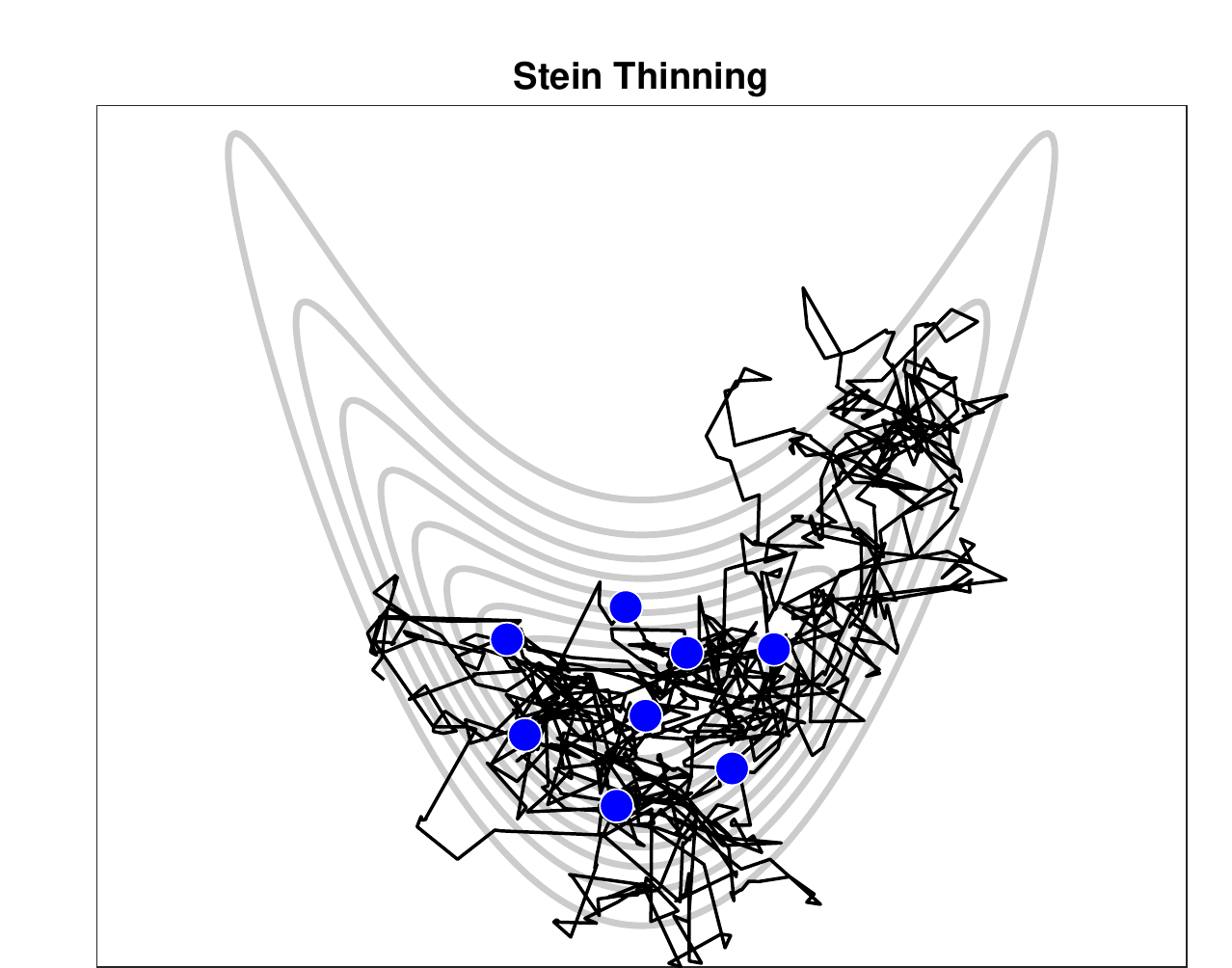}
\caption{\textit{Stein points} (left) \textit{and Stein thinning} (right)\textit{:} 
Stein points are generated by sequential (greedy) minimisation of \ac{ksd} to select $n=8$ states $\{\bm{x}_i\}_{i=1}^n$ in an approximation $Q_n$ to $P$.  The numbers indicate the order in which the states $\bm{x}_i$ were selected.
Stein thinning restricts the continuous inner-loop optimisation problem in Stein points to a discrete search over a \ac{mcmc} sample path (black).
}
\label{fig: sp}
\end{figure}

\begin{remark}
The Gaussian distribution $P$ in \Cref{ex: Stein points} is analytically tractable, so \ac{ksd} is not actually needed and \ac{mmd} can be used.
You are encouraged to extend \Cref{ex: Stein points} by replacing $P$ with an intractable posterior distribution arising in an application of Bayesian statistics.
\end{remark}

\begin{remark}
Sequential greedy optimisation provides one practical solution to part (c) of \Cref{ex: Stein points}, and was studied in detail in \citet{chen2018stein,chen2019stein}.
States $\{\bm{x}_i\}_{i=1}^n$ constructed in this way were called \emph{Stein points}.
\end{remark}

\subsection{Optimal Approximation}

In challenging applications of Bayesian statistics, the optimisation over $\mathcal{X}$ that was required to perform \Cref{ex: Stein points} will be difficult.
Nevertheless, approximate sampling from $P$ may still be possible using \ac{mcmc} or \ac{smc}.
Through the optimisation of weights, Stein discrepancy provides an means to improve the approximations produced by \ac{mcmc} or \ac{smc}, in a similar spirit to how importance sampling is sometimes used.

However, if we were to apply \Cref{lem: opt weights} with the kernel $k_P$ in place of $k$ we would obtain a degenerate solution, since $z_i = \int k_P(\cdot, \bm{x}_i) \mathrm{d}P = 0$ and optimal weights are $\bm{w} = \bm{0}$.
This makes sense, since we know that all $f \in \mathcal{H}(k_P)$ integrate to 0.
In order to make progress we need an additional constraint on the weights, and for this purpose it is natural to impose that $w_1 + \dots + w_n = 1$.
This leads to the following extension of \Cref{lem: opt weights}, which we present for a Stein kernel:

\begin{lemma} \label{lem: opt weights stein}
Let $\bm{x}_1,\dots,\bm{x}_n \in \mathcal{X}$ be distinct.
The optimal weights 
\begin{align*}
\argmin_{\substack{\bm{w} \in \mathbb{R}^n \\ \bm{1}^\top \bm{w} = 1}} D_{k_P}\left(\sum_{i=1}^n w_i \delta(\bm{x}_i) \right)
\end{align*}
are 
\begin{align*}
\bm{w} & = \frac{\bm{K}_P^{-1} \bm{1}}{\bm{1}^\top \bm{K}_P^{-1} \bm{1}}
\end{align*}
where $[K_P]_{ij} = k_P(\bm{x}_i,\bm{x}_j)$.
\end{lemma}
\begin{proof}
From \Cref{lem: ksd explicit} we have 
$$
D_{k_P}\left(\sum_{i=1}^n w_i \delta(\bm{x}_i) \right)^2 = \bm{w}^\top \bm{K}_P \bm{w}, 
$$
so the optimisation problem is
\begin{align*}
\argmin \bm{w}^\top \bm{K}_P \bm{w} \qquad \text{s.t.} \qquad \bm{1}^\top \bm{w} = 1 .
\end{align*}
This can be solved using the method of Lagrange multipliers to obtain the stated result.
\end{proof}

\noindent See the left panel of \Cref{fig: bbis}. 
As with optimal weights for \ac{mmd}, the linear system which must be solved can be numerically ill-conditioned when $n$ is large, or when two of the $\bm{x}_i$ are very close or identical.
Techniques for numerical regularisation could be used.

\begin{exercise}[Bias correction with KSD] \label{exam: bias correct}
Consider again the distribution $P$ and kernel $k$ from \Cref{ex: Stein points}.
\begin{enumerate}[label=(\alph*)]
\item Generate and store a sequence $(\bm{x}_i)_{i=1}^{100}$ of independent samples from $P$.
\item For each $n = 1,\dots,100$, calculate and plot the values of
$$
D_k\left(P, \frac{1}{n} \sum_{i=1}^n \delta(\bm{x}_i) \right)
\qquad 
\text{and}
\qquad
D_k\left(P, \sum_{i=1}^n w_i^{(n)} \delta(\bm{x}_i) \right)
$$
where $\bm{w}^{(m)} = (w_1^{(m)},\dots,w_n^{(m)})$ are the optimal weights from \Cref{lem: opt weights stein}, based on the states $\{\bm{x}_i\}_{i=1}^m$.
What do you observe?
\item What happens if, instead, we consider $n = 1,\dots,1000$?
\end{enumerate}
\end{exercise}

\begin{figure}
\centering
\includegraphics[width = 0.49\textwidth,clip,trim = 4cm 10.5cm 4cm 9cm]{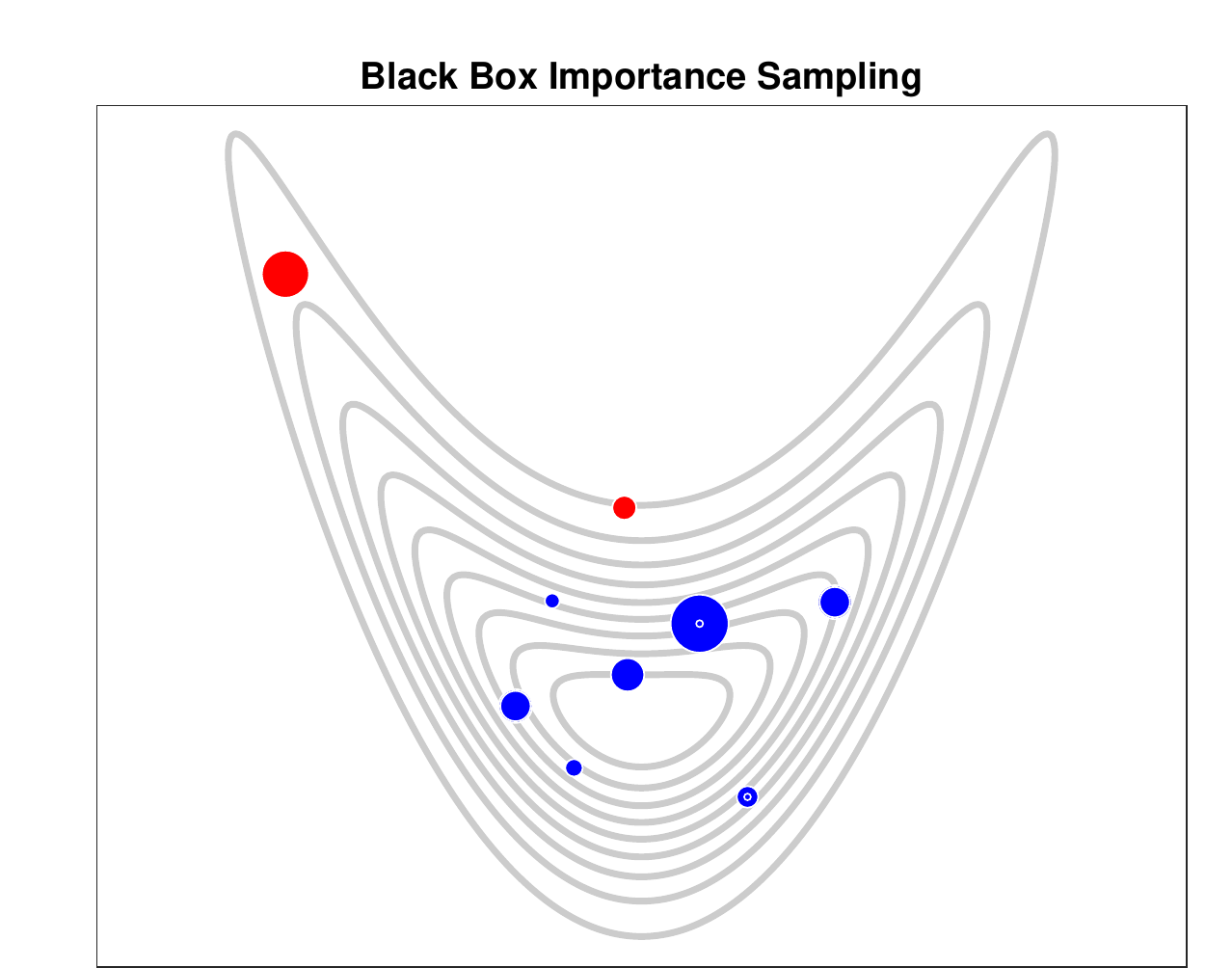}
\includegraphics[width = 0.49\textwidth,clip,trim = 4cm 10.5cm 4cm 9cm]{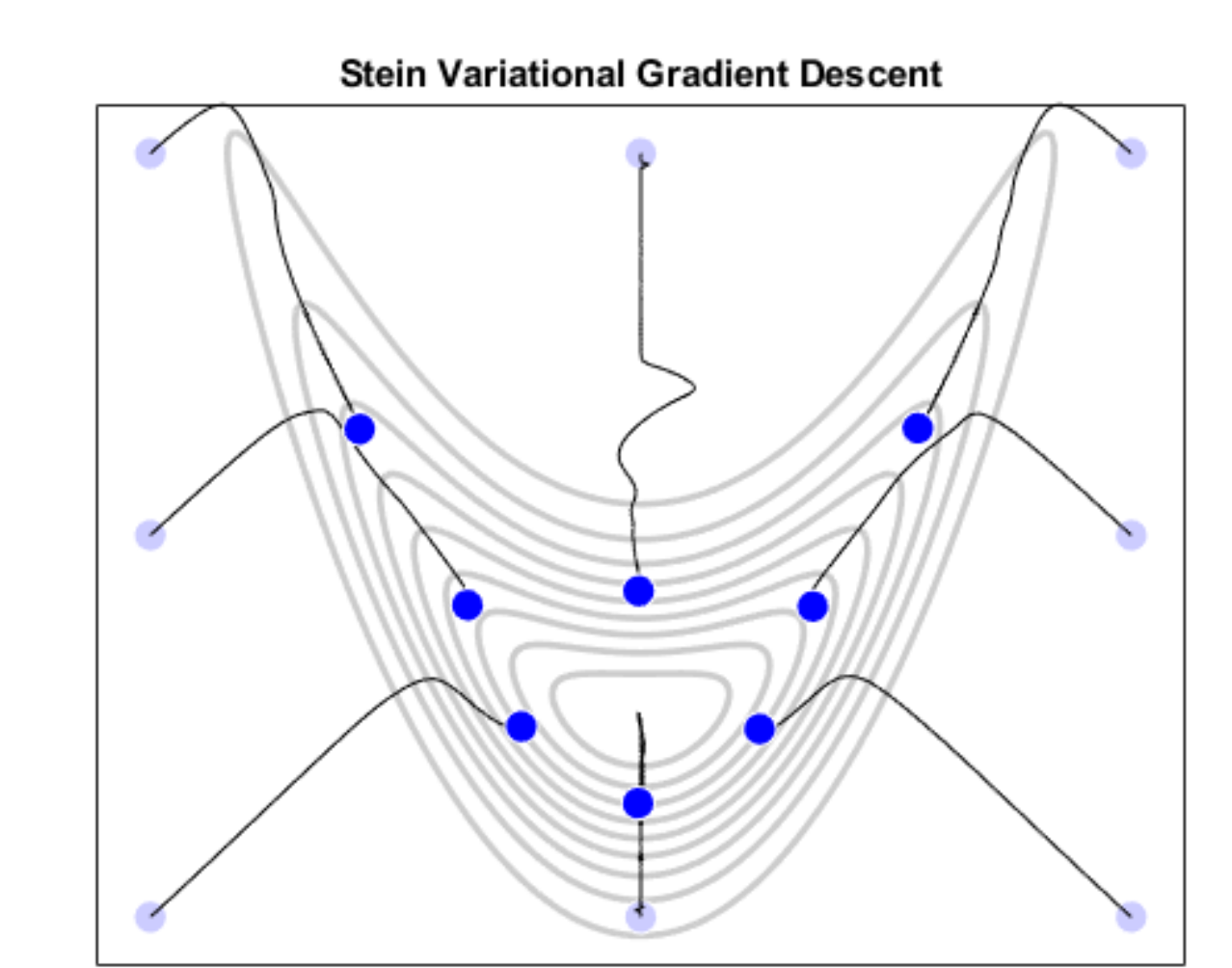}
\caption{\textit{Black box importance sampling} (left) \textit{and Stein variational gradient descent} (right)\textit{:}
In black box importance sampling the weights $w_1,\dots,w_n$ are obtained by minimising \ac{ksd} in the manner of \Cref{lem: opt weights stein}.
Blue indicates states $\bm{x}_i$ with positive weights $w_i>0$, while red indicates negative weights $w_i < 0$.
The size of the circles is proportional to $|w_i|$.
Stein variational gradient descent \citep{liu2016stein} is one of a number of recently developed algorithms based on \ac{ksd}; here an initial discrete distribution (light blue circles) are evolved in time (black lines) toward an optimal quantisation (dark blue circles) of $P$.
}
\label{fig: bbis}
\end{figure}

Remarkably, the use of Stein discrepancy in \Cref{exam: bias correct} can provide consistent approximations of $P$ even if the states $\bm{x}_i$, generated in step (a), are not sampled from $P$ (provided, at least, that they are sampled from a distribution that is not \textit{too} different from $P$).
See \citet{liu2017black,hodgkinson2020reproducing} and Theorem 3 of \citet{riabiz2020optimal}.
This is reminiscent of \textit{importance sampling}, and \citet{liu2017black} termed it \textit{black-box importance sampling}.

\begin{exercise}[Pathologies of KSD] \label{exer: pathologies}
The purpose of this final exercise is to illustrate one of the main weaknesses of \ac{ksd}; insensitivity to distant high probability regions.
Consider the Gaussian mixture model
$$
p(x) \propto \exp(-x^2) + \exp(-(x-c)^2)
$$
for $c \geq 0$.
\begin{enumerate}[label=(\alph*)]
\item Calculate (analytically) the gradient $(\nabla \log p)(x)$. 
\item Derive the kernel $k_P(x,y)$, using $k(x,y) = (1 + (x-y)^2)^{-1/2}$ (to ensure convergence control; c.f. \Cref{thm: ksd cc}).
\item Generate and store an independent sample $(x_i)_{i=1}^n$ from $\mathcal{N}(0,1)$, with $n = 100$.
\item Plot $D_{k_P}(\frac{1}{n} \sum_{i=1}^n \delta(x_i))$ as a function of $c \in (0,10)$.
\item What do you conclude about the sensitivity of \ac{ksd} to distant high probability regions?
\end{enumerate}
\end{exercise}

\paragraph*{Chapter Notes}
The canonical Stein operator is sometimes called the \emph{Langevin--Stein operator} due to its close connection with the generator of a Langevin diffusion process \citep{barbour1988stein,barbour1990stein,gorham2015measuring}.
The general concept of a Stein discrepancy, in \Cref{lem: Stein discrep}, was introduced in \citet{gorham2015measuring}.
The Stein kernel was introduced in \citet{oates2017control} and \ac{ksd} was later introduced simultaneously in \citet{chwialkowski2016kernel,liu2016kernelized}.
\Cref{lem: Stein kernel} can be found in \citet{south2020semi}.
\Cref{thm: ksd cc} was slightly generalised to allow for invertible linear transformations of $\bm{x}$ in the kernel $k$ in Theorem 4 of \citet{chen2019stein}.
The optimal weights in \Cref{lem: opt weights stein} are numerically ill-conditioned when $n$ is large; to address this, \citet{riabiz2020optimal} showed that greedy subset selection can be almost as accurate, but with lower computational complexity.
Stein discrepancy is an active research topic at the moment, with many extensions attracting attention, such as to non-Euclidean domains \citep{barp2018riemann}, to other function classes besides kernels \citep{si2020scalable,grathwohl2020learning}, and inspiring new algorithms such as \textit{Stein variational gradient descent} \citep{liu2016stein} (see the right panel of \Cref{fig: bbis}).
For a recent literature review, see \citet{anastasiou2021stein}.

\footnotesize

\bibliographystyle{plainnat}

\normalfont

\newpage

\section{Partial Solutions to Exercises}

This section contains solutions to analytic parts of the exercises in the main text.
For the numerical part of the exercises, Matlab solutions are provided.

\subsection{\Cref{ex: Gaussian kernel cubature}}

\paragraph{(a)}

In dimension $d$ we have the kernel mean embedding
$$
\mu_P(\bm{x}) = \left( \frac{\sigma^2}{2 + \sigma^2} \right)^{d/2} \exp\left\{ - \frac{1}{(2 + \sigma^2)} \|\bm{x}\|^2 \right\} .
$$

\subsection{\Cref{exer: rates mmd}}

\paragraph{(a)}

The maximum value operator $z \mapsto z_+$ that appears in the kernels of \Cref{ex: Sobolev kernels} is irrelevant when restricting to $\mathcal{X} = [0,1]$, since $1 - |x-y| \geq 0$ for all $x,y \in [0,1]$.
Thus we consider the following kernels on $\mathcal{X} = [0,1]$:
\begin{center}
\begin{tabular}{|l|l|} \hline
$k(x,y)$ \; \; {\normalfont ($r = |x-y|$, $x,y \in [0,1]$)} & {\normalfont order} \\ \hline
$(1 - r )$ & $s=1$ \\
$(1 - r)^3 (3r+1)$  & $s=2$  \\
$(1 - r)^5 (8r^2 + 5r + 1)$ &  $s=3$ \\ \hline
\end{tabular}
\end{center}

\noindent Splitting $\int_0^1 k(x,y) \mathrm{d}y$ into a sum of $\int_0^x k(x,y) \mathrm{d}{y}$ and $\int_x^1 k(x,y) \mathrm{d}y$, the integrals become straight forward and we can evaluate them analytically:

\begin{center}
\begin{tabular}{|l|l|l|} \hline
{\normalfont order} & $\mu_P(x)$ & $\iint k(x,y) \mathrm{d}x \mathrm{d}y$ \\ \hline
$s=1$ & $-x^2 + x + \frac{1}{2}$ & $\frac{2}{3}$ \\
$s=2$  & $x^4 - 2x^3 + x + \frac{2}{5}$ & $\frac{3}{5}$  \\
$s=3$ &  $-2x^8+8x^7-\frac{35}{3}x^6 + 7x^5 - \frac{7}{3}x^3 + x + \frac{1}{3}$ & $\frac{19}{36}$ \\ \hline
\end{tabular}
\end{center}

\subsection{\Cref{ex: Stein points}}

\paragraph{(b)}

Let $u(\bm{x}) = (\nabla \log p)(\bm{x})$ and consider the kernel $k(\bm{x},\bm{y}) =  (1 + \|\bm{x}-\bm{y}\|^2 )^{-1/2}$.
For the case where $P = \mathcal{N}(\bm{0},\bm{I})$ we have $u(\bm{x}) = -\bm{x}$.
Then we compute
\begin{align*}
	\nabla_{\bm{x}} k(\bm{x},\bm{y}) & = - \frac{1}{ \left(1 + \|\bm{x}-\bm{y}\|^2 \right)^{3/2} } (\bm{x}-\bm{y}) \\
	\nabla_{\bm{y}} k(\bm{x},\bm{y}) & = \frac{1}{ \left(1 + \|\bm{x}-\bm{y}\|^2 \right)^{3/2} } (\bm{x}-\bm{y}) \\
	\nabla_{\bm{x}} \cdot \nabla_{\bm{y}} k(\bm{x},\bm{y}) & = - \frac{3 \|\bm{x}-\bm{y}\|^2}{ \left(1 + \|\bm{x}-\bm{y}\|^2 \right)^{5/2} } + \frac{d}{ \left(1 + \|\bm{x}-\bm{y}\|^2 \right)^{3/2} } 
\end{align*}	
which leads to
\begin{align*}
	k_P(\bm{x},\bm{y}) & := \nabla_{\bm{x}} \cdot \nabla_{\bm{y}} k(\bm{x},\bm{y}) + [\nabla_{\bm{x}} k(\bm{x},\bm{y})]^\top u(\bm{y}) + [\nabla_{\bm{y}} k(\bm{x},\bm{y})]^\top u(\bm{x}) + k(\bm{x},\bm{y}) [u(\bm{x})^\top u(\bm{y})] \\
	& = - \frac{3 \|\bm{x}-\bm{y}\|^2}{ \left(1 + \|\bm{x}-\bm{y}\|^2 \right)^{5/2} } + 2 \beta \left[ \frac{ d + [u(\bm{x}) - u(\bm{y})]^\top (\bm{x}-\bm{y}) }{ \left(1 + \|\bm{x}-\bm{y}\|^2 \right)^{3/2} } \right] + \frac{ u(\bm{x})^\top u(\bm{y}) }{ \left(1 + \|\bm{x}-\bm{y}\|^2 \right)^{1/2} }
\end{align*}

\paragraph{(e)}

Theorem 5 of \citet{gorham2017measuring} proves that the Stein kernel $k_P$ based on the Gaussian kernel $k$ provides weak convergence control when $P$ is distantly dissipative (and $\nabla \log p$ is Lipschitz, under \Cref{ass: lipz}).
However, for $d \geq 3$, Theorem 6 of \citet{gorham2017measuring} proves that the corresponding \ac{ksd} does \textit{not} provide weak convergence control.

\subsection{\Cref{exer: pathologies}}

\paragraph{(a)}

The required gradient is
\begin{align*}
(\nabla \log p)(x) & = - 2 \left[ \frac{x \exp(-x^2) + (x-c) \exp(-(x-c)^2)}{\exp(-x^2) + \exp(-(x-c)^2)} \right] .
\end{align*}

\end{document}